\documentclass[acmsmall]{acmart}

\AtBeginDocument{%
  }

\copyrightyear{2025}
\acmYear{2025}
\acmDOI{XXXXXXX.XXXXXXX}

\usepackage{float}
\usepackage{multirow}
\usepackage{mathpartir}
\usepackage{fancyvrb}
\usepackage{xcolor}
\usepackage{subcaption}
\usepackage{caption}
\usepackage{listings}
\usepackage{tikz}
\usetikzlibrary{shapes, decorations.pathmorphing}
\usepackage{syntax}
\usepackage{matlab-prettifier}
\usepackage{adjustbox}
\usepackage{xspace}
\usepackage{enumitem}
\usepackage{comment}
\usepackage{amsmath, bm}
\usepackage{wrapfig}
\usetikzlibrary{positioning}
\usetikzlibrary{calc}
\usetikzlibrary{fit}
\usepackage{relsize}
\usepackage{xfrac}
\usepackage{tikz}
\usepackage{booktabs} % Required for better table formatting
\usetikzlibrary{shapes.geometric, arrows}
\usepackage{algorithm}
\usepackage[noend]{algpseudocode}

\definecolor{backcolour}{RGB}{250, 250, 250}
\newtheorem{definition}{Definition}[section]  % Separate numbering for definitions
\newtheorem{lemma}{Lemma}[section]            % Separate numbering for lemmas
\newtheorem{theorem}{Theorem}[section]            % Separate numbering for lemmas
\newtheorem{proposition}{Proposition}[section]            % Separate numbering for lemmas

\newcommand{\rmv}[1] {}
\newcommand{\s}[1] {\mathsf{#1}}
\newcommand{\myparagraph}[1] {\subsubsection*{\textbf{#1}}}
\newcommand{\set}[1] {\{#1\}}
\newcommand{\setpred}[2] {\set{#1\: |\: #2}}
\newcommand{\powerset}[1] {2^{#1}} %%%%% Remove \power later in the file
\newcommand{\RNum}[1]{\uppercase\expandafter{\romannumeral #1\relax}}

\newcommand{\relcomp} {\circ}
\newcommand{\id}[1] {\s{id}_{#1}}
\newcommand{\pgm} {\s{P}}
\newcommand{\looppgm} {\s{L}}
\newcommand{\loopfree} {\alpha}
\newcommand{\inv} {\s{I}}
\newcommand{\invar} {\inv}
\newcommand{\annopgm} {\s{A}}
\newcommand{\alooppgm} {\s{A}_{\s{L}}}

\newcommand{\prog}{\pgm}
\newcommand{\reach}{\mathcal{R}}

\newcommand{\tempr}{\ensuremath{\mathcal{F}}}
\newcommand{\tempi}{\ensuremath{\mathcal{I}}}
\newcommand{\meet}{\ensuremath{\wedge}}

\newcommand{\strctless}{\ensuremath{\prec}}
\newcommand{\less}{\ensuremath{\preceq}}
\newcommand{\func}{\ensuremath{\Pi}}

\definecolor{syntax}{rgb}{0.0,0.0,1.0}
\definecolor{darkgreen}{rgb}{0.0, 0.5, 0.0}
\newcommand{\checkr}{\texttt{\textcolor{syntax}{\textbf{\small{check}}}}}
\newcommand{\checki}{\texttt{\textcolor{syntax}{\textbf{\small{checkInv}}}}}
\newcommand{\rp}{\mathsf{t}}
\newcommand{\rques}{\mathsf{t}^?}
\newcommand{\gc}{\texttt{\textcolor{syntax}{\textbf{\small{getCandidate}}}}}
\newcommand{\gtraces}{\texttt{\textcolor{syntax}{\textbf{\small{getTraces}}}}}
\newcommand{\gi}{\texttt{\textcolor{syntax}{\textbf{\small{getInv}}}}}

\newcommand{\fr}{\texttt{\textcolor{syntax}{\textbf{\small{findRanking}}}}}
\newcommand{\gcounter}{\texttt{\textcolor{syntax}{\textbf{\small{getCex}}}}}
\newcommand{\gcounterinv}{\texttt{\textcolor{syntax}{\textbf{\small{getCexInv}}}}}
\newcommand{\refine}{\texttt{\textcolor{syntax}{\textbf{\small{refine}}}}}
\newcommand{\prefnum}{\texttt{\textcolor{darkgreen}{\ensuremath{\bm{\mathtt{P}_{ref}}}}}}
\newcommand{\prefiter}{\texttt{\textcolor{darkgreen}{\ensuremath{\bm{\mathtt{P}_{iter}}}}}}
\newcommand{\ptraces}{\texttt{\textcolor{darkgreen}{\ensuremath{\bm{\mathtt{P}_{traces}}}}}}

\newcommand{\lati}{\mathcal{L}}

\newcommand{\true}{\ensuremath{\mathsf{true}}\xspace}
\newcommand{\false}{\ensuremath{\mathsf{false}}\xspace}
\newcommand{\rf}{\ensuremath{f}}

\newcommand{\init}{\ensuremath{Init}\xspace}
\newcommand{\ini}{\init}

\newcommand{\stat}{\ensuremath{S}\xspace}

\newcommand{\dom}{\mathsf{dom}}

\newcommand{\anno}{\annopgm}
\newcommand{\pro}{\pgm}

\newcommand{\bracs}[1]{\llbracket #1 \rrbracket}
\newcommand{\core}[1]{\mathcal{C}(#1)}

\newcommand{\toolname}{\textsc{Syndicate}\xspace}

\definecolor{dkgreen}{rgb}{0,0.6,0}
\definecolor{gray}{rgb}{0.5,0.5,0.5}
\definecolor{mauve}{rgb}{0.58,0,0.82}

\newcounter{number}
\newcommand{\mycounter}{...(\thenumber) \stepcounter{number}}


\begin{document}

\title{Efficient Ranking Function-Based Termination Analysis with Bi-Directional Feedback}

\author{Yasmin Chandini Sarita}
\affiliation{%
  \institution{University of Illinois Urbana-Champaign}
  \country{USA}
}
\email{ysarita2@illinois.edu}

\author{Avaljot Singh}
\email{avaljot2@illinois.edu}
\affiliation{%
  \institution{University of Illinois Urbana-Champaign}
  \country{USA}
}

\author{Shaurya Gomber}
\affiliation{%
  \institution{University of Illinois Urbana-Champaign}
  \country{USA}
}
\email{sgomber2@illinois.edu}

\author{Gagandeep Singh}
\affiliation{%
  \institution{University of Illinois Urbana-Champaign}
  \country{USA}
}
\email{ggnds@illinois.edu}

\author{Mahesh Vishwanathan}
\affiliation{%
  \institution{University of Illinois Urbana-Champaign}
  \country{USA}
}
\email{vmahesh@illinois.edu}

\begin{abstract}
Synthesizing ranking functions is a common technique for proving the termination of loops in programs. A ranking function must be bounded and decrease by a specified amount with each iteration for all reachable program states. However, the set of reachable program states is often unknown, and loop invariants are typically used to overapproximate this set. So, proving the termination of a loop requires searching for both a ranking function and a loop invariant. Existing ranking function-based termination analysis techniques can be broadly categorized as (i) those that synthesize the ranking function and invariants independently, (ii) those that combine invariant synthesis with ranking function synthesis into a single query, and (iii) those that offer limited feedback from ranking function synthesis to guide invariant synthesis. 
These approaches either suffer from having too large a search space or inefficiently exploring the smaller, individual search spaces due to a lack of guidance. 
% All of these approaches suffer from inefficiencies. 

In this work, we present a novel termination analysis framework \toolname, which exploits bi-directional feedback to guide the searches for both ranking functions and invariants in tandem that constitute a proof of termination for a given program. This bi-directional guidance significantly enhances the number of programs that can be proven to terminate and reduces the average time needed to prove termination compared to baselines that do not use the bi-directional feedback. The \toolname framework is general and allows different instantiations of templates, subprocedures, and parameters, offering users the flexibility to adjust and optimize the ranking function synthesis. 

We formally show that, depending on the templates used, \toolname is both relatively complete and efficient, outperforming existing techniques that achieve at most one of these guarantees. Notably, compared to more complex termination analysis methods extending beyond a single ranking function synthesis, \toolname offers a simpler, synergistic approach. Its performance is either comparable to or better than existing tools in terms of the number of benchmarks proved and average runtime. 
% Additionally, due to their orthogonal approaches, future combinations of \toolname with existing techniques could yield even better results.

\end{abstract}


\maketitle

\section{Introduction}
\label{sec:intro}
Termination analysis is one of the most challenging components of program verification. It is used for verifying systems code~\cite{terminator} and within software model checking~\cite{completelin}. 
Several techniques for automatically proving termination have been explored, including using various reduction orders, finding recurrent sets, fixpoint logic solving, etc.~\cite{aprove, ultimate, dynamite, popl23}. One of the most common techniques to prove the termination of programs and the focus of this work is finding \textit{ranking functions}.  Ranking functions have the property that for all reachable program states, they are bounded from below and decrease by at least a certain amount in each iteration. Finding ranking functions is non-trivial because computing the exact set of reachable states is undecidable. To overcome this, loop invariants are used to over-approximate the reachable states. So, to prove the termination of a program using this technique, there are two unknowns - the ranking function and the loop invariant. 

In the context of classical property-guided synthesis~\cite{propinvar,pdr}, synthesizing a proof of termination, which requires both a ranking function and an invariant, is challenging. The difficulty arises because when synthesizing either component, we do not have access to the other component, which would act as the target property. For instance, when synthesizing an invariant, our goal is to create one that is useful for proving termination. However, since we don't yet have the ranking function, we can't assess whether the synthesized invariant is useful in this context, as the ranking function would be needed to evaluate its effectiveness in proving termination. So, finding ranking functions and invariants to prove termination is associated with the following challenges.

\textbf{Challenge 1: Dual Search. }
At a high level, existing approaches address the dual search of ranking functions and invariants through one of the following strategies:  (1) Methods like \cite{nuterm, dynamite, ultimate} synthesize the invariants and ranking functions independently, where the invariant synthesis does not know the current state of the ranking function search and vice versa. Due to the independence of the searches, these techniques frequently spend a significant amount of time discovering invariants that are either not useful to prove termination or are stronger than necessary.  (2) Techniques like ~\cite{terminator, bits, cooknew, datadrivenloop} use their current set of candidate ranking functions to guide a reachability analysis procedure. However, this procedure uses an external temporal safety checker, so insights gained during one call to the safety checker cannot be used to guide the ranking function generation procedure or future calls to the safety checker. (3) Methods such as \cite{chc,popl23,togethersmt} formulate both searches as a composite query. These techniques pose the synthesis of ranking functions and an over-approximation of the reachable states as the task of verifying the satisfiability of a combined formula. This couples both searches tightly, making the search space for the combined query significantly large, being a product of the individual search spaces.

\textbf{Challenge 2: Efficiency \& Completeness. }
Achieving efficiency and completeness while handling a wide variety of ranking functions and invariants is a significant challenge. While enumerating all possible candidate ranking functions and invariants can trivially ensure completeness, it often leads to prohibitively large runtimes, making such approaches impractical for real-world applications. Many existing techniques focus on specific classes of ranking functions~\cite{polykincaid, completelin, completelin1}, such as linear or polyhedral ranking functions, providing completeness guarantees. However, this limits the kind of programs and ranking functions they can handle. While some methods offer completeness for a broader set of ranking functions~\cite{chc,popl23}, they still do not provide theoretical guarantees for efficient performance. Striking the right balance between efficiency and completeness while handling complex programs and ranking functions is a challenge in termination analysis.

\textbf{Key Idea: Bi-directional Feedback. }
We argue that for a given program, although the search for ranking functions and the corresponding invariants must not be agnostic of each other, directly encoding them as a monolithic query can make each step of the search expensive. We propose that they should still be searched using separate subroutines, but they should provide feedback to each other in the following manner: (1) Counterexamples to the ranking function validation check that are not reachable program states can guide the invariant synthesizer to generate invariants that exclude these counterexamples. On the other hand, when the counterexamples are reachable and no invariant can exclude them, this search for possible invariants aids in discovering new conclusively reachable program states, thereby supporting the invariant synthesis process. (2) Provably reachable counterexamples found by an invariant synthesizer can guide the ranking function synthesizer to ensure that future potential ranking functions are bounded and reducing on these counterexamples. (1) and (2) constitute a \textit{bi-directional feedback strategy} that enables an efficient synergistic search for a ranking function and an invariant just \textit{strong-enough} to prove termination. This approach is efficient because it maintains smaller, individual search spaces for ranking functions and invariants while leveraging bi-directional feedback to enable faster convergence towards a valid termination proof.

Further, for programs with multiple loops, the bi-directional feedback becomes even more important because the termination proofs across different loops can neither be searched independently due to excessive imprecision nor modeled as a single query due to the prohibitively large product search space. In nested loops, the invariant of the outer loop defines the over-approximation of the initial set of its inner loop. Conversely, the inner invariant describes the over-approximated transition relation of the outer loop's body. Similarly, the searches for termination proofs of sequential loops affect each other. Analyzing multiple loops is more efficient if the ranking function search and the invariant search for all loops guide each other using bi-directional feedback.

\textbf{Our framework: \toolname. }
Based on our key idea, we introduce a new, general framework called \toolname that enables the synergistic synthesis of ranking functions and invariants for efficient termination analysis of complex programs containing arbitrarily nested loops, disjunctive conditionals,  and non-linear statements. It is parameterized by a set of possible invariants $\tempi$ and a set of possible ranking functions $\tempr$ and can be instantiated with different choices to obtain different termination analyses. For each loop, counterexamples from the invariant search guide the ranking function search, and counterexamples from the ranking function search guide the invariant search. Further, the invariant and ranking function searches for each loop are guided by the analyses of the other loops in the program. Based on the choices for $\tempi$ and $\tempr$, \toolname is relatively complete, meaning that if complete procedures exist for the template-specific functions in our framework, then \toolname will either find a ranking function and an invariant within their respective templates to prove termination (if they exist) or will terminate, indicating that no such ranking function and invariant can be found. We show that \toolname achieves this guarantee with greater efficiency than other approaches lacking completeness guarantees.

\textbf{Main Contributions}:
\begin{itemize}
    \item We introduce a novel bi-directional feedback strategy for efficiently synthesizing ranking functions and invariants to prove termination automatically. 
    \item We present a new, general termination analysis framework, \toolname, based on bi-directional feedback, which can be instantiated with different templates for ranking functions and invariants and can handle complex programs with arbitrarily nested and sequential loops, disjunctive conditionals, and non-linear statements. 
    \item We show an additive upper bound on the number of iterations in terms of the depths of lattices of invariants and ranking functions while guaranteeing relative completeness of \toolname for a subset of possible templates (\S~\ref{sec:efficiency}). 
    \item We provide an implementation of \toolname (code in the supplementary material) and evaluate it on challenging benchmarks~\cite{svcomp,termcomp,nuterm} (\S~\ref{sec:evaluation}). Our method is already comparable or even better than existing state-of-the-art termination analysis tools—some of which go beyond ranking function synthesis or optimize SMT solvers~\cite{nuterm, aprove, dynamite, ultimate, verymax}. Notably, we prove the termination of benchmarks that none of these advanced tools can prove.
\end{itemize}

\section{Overview}
\label{sec:overview}
We aim to find a ranking function $\rf$ for a loop program $\looppgm = \texttt{while}(\beta) \texttt{ do }\pgm_1 \texttt{ od}$ in some program $\pgm$. We use $\reach$ to refer to the set of reachable program states at the beginning of the loop and $\bracs{\pgm_1}$ to refer to the relation that captures the semantics of the loop body $\pgm_1$. Note that $\pgm_1$ can itself have loops. A function $\rf$ is a valid ranking function for the loop program $\looppgm$ w.r.t $\reach, \bracs {\pgm_1}$ if it is bounded for all states in $\reach$ and reduces for all pairs of states in $\bracs {\pgm_1}$. In this work, we choose $\epsilon=0, \delta=1$. However, this can be generalized (Lemma~\ref{aplemma:rfgen} in Appendix~\ref{appendix:proofs}). 
\begin{enumerate}
    \item \textit{Bounded: } $\forall s \in \reach \cdot f(s) \geq \epsilon$ 
    \item \textit{Reducing: } $\forall (s, s') \in \pgm_1 \cdot f(s) - f(s') \geq \delta$
\end{enumerate}

\textbf{Algorithm State. }
At each step of our algorithm, we maintain a state represented by a tuple \(\langle t, \alooppgm \rangle\). The first component, \(t\), is a set of pairs of states at the start and end of the loop body, serving as an under-approximation of the reachable states \(\reach\) and the loop body transition relation, i.e., \(\dom(t) \subseteq \reach\) and \(t \subseteq \bracs{\pgm_1}\). The second component, \(\alooppgm\), stores the program along with invariants, $\inv$, for each loop, which provides an over-approximation of the reachable states and transition relations, i.e., \(\reach \subseteq \inv\) and \(\bracs{\pgm_1} \subseteq \inv \times (\inv \cap \llbracket \neg\beta \rrbracket)\). The set \(t\) can be initialized as an empty set or by executing the program on random inputs, while loop invariants are initialized to the trivial invariant, \(\stat\) (the set of all states, i.e., $\inv_{true})$. We define \(\alooppgm = \invar \; @ \; \texttt{while}(\beta) \texttt{ do } \anno_1 \texttt{ od}\) as the \textit{annotated program} corresponding to $\looppgm$, where each loop is annotated with its invariant, and the semantics of \(\looppgm\) can be over-approximated by the invariants in \(\alooppgm\), denoted by \(\llbracket \alooppgm \rrbracket\) (formally defined in \S~\ref{sec:annotatedpgm}).

\textbf{\toolname Framework. }
We introduce \toolname (Fig.~\ref{fig:overview}) as interleaved Generate and Refine phases that interact through bi-directional feedback. The Generate Phase generates a new candidate ranking function, while the Refine phase tries to refine the algorithm state to better approximate the set of reachable states $\reach$ and the loop body transition relation $\bracs{\pgm_1}$. The Refine phase iteratively generates and checks new candidate invariants, but is different from typical counterexample guided synthesis because instead of running until a target property is proved, the refine phase guides the search for the target property, the ranking function. For simplicity, we show \toolname for a single loop, so instead of $t$ and $\annopgm$, we can represent that state as $\langle r, \inv\rangle$, where $r$ is a set of states at the beginning of the loop and $\inv$ is the invariant of the loop. More concretely, after generating a candidate ranking function $\rf$ in the Generate phase, if $\rf$ is invalid, we get a counterexample, $p$. \toolname then uses $p$ to guide the Refine phase, where new candidate invariants are generated and checked for correctness. During this process, \toolname determines whether or not $p$ is reachable, and in the process also finds other reachable states, $q$, generated as counterexamples during the Refine phase. Finally, if the invariant cannot be refined further, i.e., if $p$ is reachable, both $p$ and $q$ are used to guide the next iteration of the Generate phase for generating candidate ranking functions. Otherwise, $\inv$ is refined to exclude $p$, and $q$ is still used to guide the ranking function search. Consequently, both the ranking function search and invariant search guide each other until a termination proof is found, facilitating a more efficient termination analysis. \toolname is parameterized by $\tempi$, the set of possible invariants for a loop, and $\tempr$, the set of possible ranking functions. So, when we say a counterexample, $s$, is reachable, it means that there is no valid invariant in $\tempi$ that excludes $s$.

\begin{figure}
    \centering
    \includegraphics[width=0.7\linewidth]{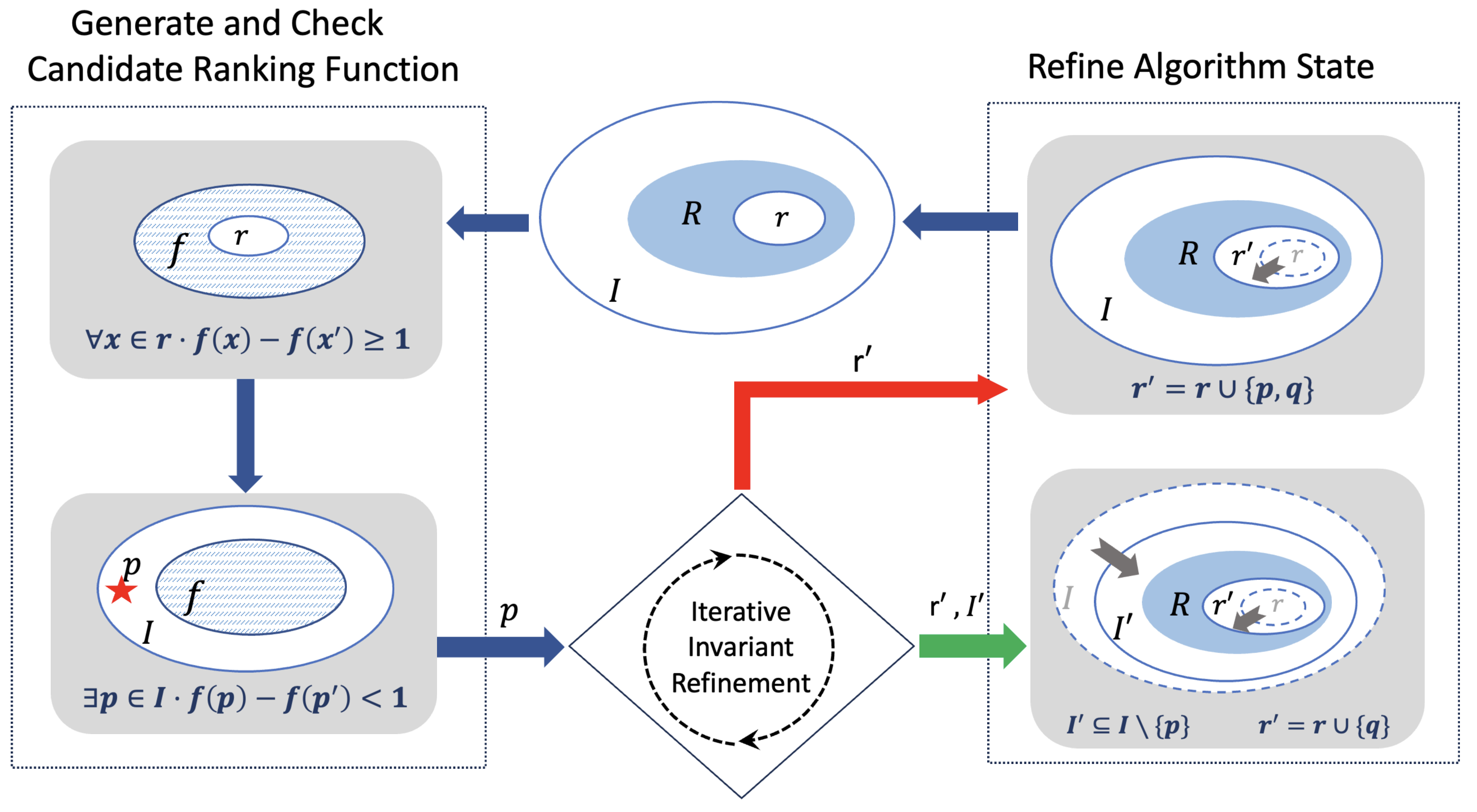}
    \caption{For a single loop, we maintain $r, \invar$ as canonical representations of $t, \alooppgm$ respectively. First, we use the set $r$ to generate a candidate ranking function and then use the invariant $\invar$ to check the validity. If we get a counterexample $p$, we try to refine the invariant $\invar$ to exclude $p$. If $p$ is not reachable, the invariant is refined successfully to $\invar'$ (green arrow). Otherwise (red arrow), we add $p$ to the set $r$. In both cases, the reachable states discovered while trying to refine $\invar$ (denoted by $q$), are also added to the set $r$.}
    \label{fig:overview}
\end{figure}

\subsection{\toolname Through an Example}
\label{sec:example}
\begin{figure}
    \begin{minipage}{0.35\textwidth}
        \centering 
        \begin{lstlisting}[numbers=left,basicstyle=\ttfamily\footnotesize,backgroundcolor=\color{backcolour}]
    y = n; 
    z = n+1;
    while (m+y >= 0 && m <= n){
        m = 2 * m + y;
        y = z;
        z = z + 1; 
    }\end{lstlisting} 
        \caption{Example terminating program}
        \label{overview:code3}
    \end{minipage}
    \hspace{0.4cm}
    \begin{minipage}{0.6\textwidth}
    \scriptsize 
        \begin{align*}
            \qquad&\rf(m', y', z', n') \\ 
            &= \max\big(1-(m'+y'), 0\big) + \max(n' - m' + 1, 0) \\ 
            &= \max\big(1-2m-y-z, 0\big) + \max(n-2m-y+1,0) \\
            &= \max\big(1-2m-y-\fbox{\textbf{(y+1)}}, 0\big) + \max\big(n-m+1-(m+y),0\big) \\ 
            % &= \max(-2m-2y, 0) + \max\big(n-m+1-(m+y),0\big) \\ 
            % &= 0 +\max\big(n-m+1-(m+y),0\big) \\
            &= \max\big(n-m+1-(m+y),0\big) \\
            &\leq \rf(m, y, z, n) - 1
        \end{align*}
        \caption{Value of Ranking function at the end of one iteration}
        \label{overview:rankingred2}
    \end{minipage}
\end{figure}
We show how \toolname works on a terminating single-loop program shown in Fig.~\ref{overview:code3}. A ranking function for this loop is $\max(n-m+1,0) + \max(1-m-y,0)$ because for all reachable states, either $n-m+1$ is reducing and $1-m-y$ stays constant or $n-m+1$ stays constant and $1-m-y$ reduces from 1 to 0. However, we need an invariant $\inv \equiv y + 1 = z$ to prove its validity (Fig.~\ref{overview:rankingred2}). We refer to the program in Fig.~\ref{overview:code3} as $\pgm$, and the loop-body (lines 4-6) as $\pgm_1$. 

As stated before, for a single loop, we can represent the state as $\langle r, \inv \rangle$. We initialize $r$ with randomly generated traces of the program. Here, $r = \{(1,2,3,4), (2,2,3,4)\}$, where a program state is a tuple ($m, y, z, n$). $\inv$ is initialized with the set of all states, i.e., $\invar = \true$.  \toolname can handle other types of ranking functions and invariants, but for this example, let $\tempr$, the set of candidate ranking functions, be functions of the form $\max(a_0 + \sum_j a_jx_j, 0) + \max(b_0 + \sum_j b_jx_j, 0)$ and $\tempi$, the set of possible invariants, be conjunctions of constraints of the form $d_0 + \sum_jd_jx_j \geq 0$, where $a_j,b_j,d_j \in \mathbb{Z}$.  We find a valid ranking function and the corresponding invariant using separate queries while exchanging key information between the two searches, enabling an efficient termination analysis.
\begin{table}
    \centering
    \scriptsize 
    \resizebox{0.9\textwidth}{!}{
    \begin{tabular}{@{}lrrrr@{}}
    \hline 
         &$r$ & $\invar$ & $\rf$ & $cex$  \\
         \hline
         1&$\{(1,2,3,4),(2,2,3,4)\}$ & $\true$ & $\max(n-m,0)$ & $(1,-1,2,4)$ \\
         \hline
         2&$\{(1,2,3,4),(2,2,3,4)\}$ & \textcolor{blue}{$y+1 \geq z$} & $\max(n-m,0)$ & $(1,-1,-1,1)$ \\
         \hline 
         \multirow{2}{*}{3}&$\{(1,2,3,4),(2,2,3,4),$ & \multirow{2}{*}{$\textcolor{blue}{(y+1\geq z) \wedge (y+1\leq z)}$} & \multirow{2}{*}{$\textcolor{blue}{\max(n-m+1,0)}$} & \multirow{2}{*}{$(1,-1,0,4)$} \\  
         &$\textcolor{blue}{(1,1,2,1)}\}$ &  &  &  \\  
         \hline
         \multirow{2}{*}{4}&$\{(1,2,3,4), (2,2,3,4), $ & \multirow{2}{*}{$(y\geq z-1) \wedge (y\leq z-1)$} & $\textcolor{blue}{\max(n-m+1,0) + }$ & \multirow{2}{*}{-}   \\
         &$(1,1,2,1),\textcolor{blue}{(1,-1,0,4)}\}$ & & $\textcolor{blue}{\max(1-m-y, 0)}$ &   \\
         \hline 
    \end{tabular}
    }
    \caption{Steps taken by \toolname to prove the termination for the program in Fig.~\ref{overview:code3}}
    % \vspace{-0.75cm}
    \label{tab:overview1}
\end{table}
The steps for our running example are explained below and also displayed in Table~\ref{tab:overview1}.
\myparagraph{Iteration 1} First, in the \textbf{Generate Phase}, we find a candidate ranking function that reduces for all states in $r$. Specifically, we find a satisfying assignment for the coefficients $a_j, b_j$ in the following query, where $f \in \tempr$.
We do not have to check the bounded condition because all of the functions in $\tempr$ are bounded by 0.
        \begin{gather*}
        \forall (m,y,z,n) \in r \cdot \big( f(m,y,z,n) - f(\bracs{\pgm_1} (m,y,z,n))\geq 1 \big)
    \end{gather*}
We solve this query using Z3~\cite{z3} and obtain $f(m,y,z,n) = \max(n-m+1,0)$. Now, we check the candidate function's validity w.r.t the current invariant $\invar = \true$. This can be done by checking the satisfiability of the following formula. Here, the condition $(m+y\geq 0) \wedge (m\leq n)$ in the left side of the implication comes from the loop guard.
\begin{gather*}
    \forall (m,y,z,n) \cdot \big((m+y\geq 0) \wedge (m\leq n)\big) \implies  \big( f(m,y,z,n) - f(\bracs{\pgm_1} (m,y,z,n))\geq 1 \big)
\end{gather*}
Solving this gives a counterexample $p_1 = (1,-1,2,4)$. Next, the algorithm state is refined in the \textbf{Refine Phase}. There are two possibilities depending on the reachability of $p_1$. Either $p_1$ is a reachable state and the ranking function is invalid, or $p_1$ is not reachable. We first try to strengthen $\inv$ to $\inv'$ s.t. $p_1 \notin \inv'$. This is done in 2 sub-phases:

\textbf{Generate Invariant Sub-Phase} We generate a candidate invariant $ \inv''$ that includes all states $s \in r$ and does not include the state $p_1$. Here, $\rf$ tries to guide the refinement of an invariant that can prove its validity. This is done by solving the following query for the coefficients $d_j$ to get $\inv'' \equiv d_0 + d_1m + d_2y + d_3z + d_4n \geq 0$ and then taking a conjunction with $\inv$, i.e., $\inv' = \inv'' \wedge \inv$.
\begin{gather*}
    (1, -1, 2, 4) \notin \invar'' \wedge \big( \forall (m,y,z,n) \in r \cdot (m,y,z,n) \in \invar'' \big) 
\end{gather*}
Solving the above query, we get $\inv'' \equiv y-z+1 \geq 0$. We now proceed to check the validity of $\inv'$.

\textbf{Check Invariant Sub-Phase}
The validity of the invariant $\invar'$ is checked by solving for a counterexample for the following query. Here, $\ini$ is the set of states before the loop at line 3 executes.
\begin{gather*}
    \big(\forall s \in \init, s \in \invar'\big) \wedge \big(\forall s \in \invar' \implies \bracs{\pgm_1}s \in \invar'\big), \\ \text{ where }\ini = \{(m,y,z,n) \ | \ (y=n) \wedge (z=n+1)\}
\end{gather*}
We do not get a counterexample, so $y-z+1 \geq 0$ is a valid invariant. 
We update the algorithm state by changing $\inv$ to $\inv'$. Since we store $\inv'$ in the algorithm state, future iterations of the invariant generation phase in our algorithm will start with $\inv'$ instead of $\true$ when trying to determine the reachability of another counterexample. This is unlike algorithms~\cite{dynamite, terminator, bits} that rely on external solvers to find invariants, and so do not have information about the previous iterations.

\myparagraph{Iteration 2} 
As $r$ is unchanged, without generating a new ranking function, we check the validity of the ranking function w.r.t the new invariant.
We get a new counterexample $p_2 = (1,-1,-1,1)$. 

\textbf{Generate Invariant Sub-Phase} We try to strengthen our invariant $y + 1 \geq z$ to exclude $p_2$ by generating a candidate invariants guided by $p_2$, as well as $r$, thus generating the candidate $m-n-1 \ge 0$. 

\textbf{Check Invariant Sub-Phase} When checking the validity of the new invariant, we find the counterexample $q_2^1 = (1,1,2,1)$.  
Since $q_2^1 \in \init$, it is reachable. So, we can add it to $r$ to improve the under-approximation of the reachable states. $q_2^1$ has a property that the other states in $r$ do not, i.e., $m=n$. This guides the generation of candidate ranking functions to reduce on states where $m=n$. Since the new invariant is invalid, we go back to the Generate Invariant Phase with an additional constraint that all candidate invariants should include $q_2^1$. After further iterations of finding and checking candidate invariants, we get the invariant $y + 1 \geq z \wedge y + 1 \leq z$, which excludes $p_2$.

\myparagraph{Iteration 3} 
The set $r$ is changed from the invariant search, so we generate a new candidate ranking function. Guided by the state $q^1_2$, we get the ranking function $\max(n-m+1,0)$. We check the validity of the ranking function w.r.t. the new invariant $y + 1 \geq z \wedge y + 1 \leq z$ to get a new counterexample $p_3 = (1,-1,0,4)$. When refining the algorithm state, we are now unable to find a valid invariant that excludes $p_3$, and thus we can conclude that $p_3$ cannot be excluded from the set of reachable states using any invariant from our template. So, we add it to $r$, refining the algorithm state.

\myparagraph{Iteration 4} As the set $r$ is updated, we generate a new candidate ranking function $\rf = \max(n-m+1,0) + \max(1-m-y, 0)$ that reduces on all states in the updated $r$. We then check its validity w.r.t the current invariant $y + 1 \geq z \wedge y + 1 \leq z$. The validity check succeeds, and $\rf = \max(n-m+1,0) + \max(1-m-y, 0)$ is returned as a valid ranking function for the program.

\myparagraph{Bi-directional Feedback vs. Existing Methods} 
In iterations 1 and 2, the counterexamples found when checking the validity of candidate ranking functions guide the search for stronger, useful invariants. In iterations 2 and 3, the counterexamples found when checking the validity of candidate invariants guide the search for a valid ranking function. The bi-directional feedback makes \toolname an efficient termination analysis algorithm. Existing techniques that search for the invariant and ranking function independently~\cite{dynamite,nuterm} may waste time finding valid but irrelevant invariants such as $y \leq z$, $y \geq n$, $z \geq n + 1$ that do not assist in proving the ranking function’s validity. 

Techniques that use one direction of feedback, by providing the current ranking function counterexample in the form of a trace to an external safety checker~\cite{bits, terminator, cooknew} do not use previously found reachable ranking function counterexamples or the previous state of the invariant search to make each call to the external safety checker more efficient. Unlike these techniques, our invariant search uses $(p,p')$, $t$, and the previously found invariants to efficiently navigate the search space. Further, to the best of our knowledge, no existing technique uses intermediate counterexamples found during the invariant search, such as $q^1_2$, to guide the ranking function search. 

Lastly, some techniques completely combine the ranking function and invariant searches~\cite{verymax, chc,popl23}, but this requires navigating a much larger search space (the product of the individual spaces), whereas \toolname only needs to explore the sum of these spaces. In \S~\ref{sec:efficiency}, we show that our approach has a significantly smaller upper bound on the number of iterations compared to the combined search. In \S~\ref{sec:evaluation}, we demonstrate that this bi-directional feedback approach enables us to prove the termination of more benchmarks in a shorter amount of time than strategies that fully combine the ranking function and invariant searches and strategies that use only one direction of feedback between the two searches. We further show that \toolname can prove more benchmarks terminating than
existing state-of-the-art termination analysis tools.

\subsection{Nested Loops}
\label{sec:nestedloops}
For nested loops, we have to store $t$, which contains pairs of states and serves as an under-approximation of the transition relation of the loop. Consider the program $\pgm$ in Fig.~\ref{fig:nest}. $\pgm_i$ is the inner loop (lines 3-7), $\beta_o$ and $\beta_i$ are the loop guards at lines 2 and 3 respectively. Given a program state $s$ before executing the inner loop body, we cannot compute the state $s'$ just before line 8 because $\bracs{\pgm_i}$ is also unknown. So, we use an invariant $\invar_i$ to over-approximate $\bracs{\pgm_i}$. Similarly, we use an invariant $\invar_o$ for the outer loop. Note that the initial set for $\invar_i$ depends on $\invar_o$. Conversely, $\invar_o$ depends on $\bracs{\pgm_i}$ which is approximated using $\invar_i$. So, the invariants corresponding to all the loops in a program depend on each other. Here, we represent our algorithm state as $\langle t, (\invar_o, \invar_i) \rangle$.

\begin{wrapfigure}{l}{0.45\linewidth}
% \vspace{-0.5cm}
\begin{lstlisting}[numbers=left,basicstyle=\ttfamily\footnotesize,backgroundcolor=\color{backcolour}]
  y = n; z = n+1;
  while (y+n+1 >= k*k+z && y == z+1){
      while (m+y >= 0 && m <= n){
          m = 2*m + y;
          y = z;
          z = z+1; 
      }
      y++; z++; n--;
  }\end{lstlisting}
% \vspace{-0.5cm}
\caption{Nested loop program}
% \vspace{-0.25cm}
\label{fig:nest}
\end{wrapfigure}
We will use the same templates $\tempr$ and $\tempi$ for both loops, but this is not a requirement for our algorithm. Both invariants are initialized as the set of all states, i.e., $\invar_o = \invar_i = \true$. Using this, we compute an initial set $\init_i$ for the inner program $\pgm_i$,  $\init_i = \invar_o \cap \bracs{\beta_o} = \{(m,y,z,n) | (y+n+1 \ge k*k+z) \wedge (y=z+1)\}$. The initial set $\init_o$ for the outer while loop is $\{(m,y,z,n) | (y=n) \wedge (z=n+1)\}$. We can start by searching for a ranking function for either the outer or inner loop first. We choose the inner loop. Since the body of the inner loop is the same as the loop body in the program in Fig.~\ref{overview:code3} and $\init_i$ is a subset of the $\init$ set we computed when analyzing the single loop case, we can go through the same iterations as above to get $\rf_i = \max(n-m+1,0) + \max(1-m-y, 0) $, $\invar_i \equiv y + 1 = z$, and $\invar_o = \true$. Next, \toolname uses similar steps as discussed in \S~\ref{sec:example} to find a ranking function for the outer loop. We first generate a candidate ranking function $\rf$ similar to the single loop case. The one thing that changes while validating $\rf$ is that because we have a loop in the body of the outer loop, we cannot symbolically execute to find the program state at the end of one iteration. Instead, we define fresh variables $x'' = (m'', y'', z'', n'', k'')$ to represent the state just after the inner loop ends.  We assume $x'' \in \invar_i \wedge \neg \beta_i$. We define the validation query as:
\begin{table}
    \centering
    \resizebox{\textwidth}{!}{  
    \begin{tabular}{@{}lrrrr@{}}
    \hline 
         &$t$ & $(\invar_o,\invar_i)$ & $\rf$ & $p, p'$  \\
         \hline
         \multirow{2}{*}{1}&$\{((2,2,3,4,1),(6,7,8,0,1)),$ & \multirow{2}{*}{$(\true, y + 1 = z)$} & \multirow{2}{*}{$\max(n-k,0)$} & \multirow{2}{*}{$(1,2,4,5,5), (7,2,3,5,5)$} \\  
         &$((2,2,3,5,1), (6,8,9,0,1))\}$ &  &  &  \\  
         \hline 
         \multirow{2}{*}{2}&$\{((2,2,3,4,1),(6,7,8,0,1)),$ & \multirow{2}{*}{($\textcolor{blue}{y+1\geq z}, y+1=z)$} & \multirow{2}{*}{$\max(n-k,0)$} & \multirow{2}{*}{$(1,2,2,5,5), (8,3,4,5,5)$} \\ 
         &$((2,2,3,5,1), (6,8,9,0,1))\}$ &  &  &  \\ 
         \hline 
         \multirow{2}{*}{3}&$\{((2,2,3,4,1),(6,7,8,0,1)),$ & \multirow{2}{*}{(
         $\textcolor{blue}{(y+1\geq z) \wedge (y+1\leq z)}$, y + 1 = z)} & \multirow{2}{*}{$\max(n-k,0)$} & \multirow{2}{*}{$(1,2,3,5,5), (7,5,6,5,5)$} \\  
         &$((2,2,3,5,1), (6,8,9,0,1))\}$ &  &  &  \\  
         \hline
         \multirow{3}{*}{4}&$\{((2,2,3,4,1),(6,7,8,0,1)),  $ & \multirow{3}{*}{$((y\geq z-1) \wedge (y\leq z-1), y+1=z)$} & \multirow{3}{*}{$\textcolor{blue}{\max(n-k+1,0)}$} & \multirow{3}{*}{-}   \\
         &$((2,2,3,5,1), (6,8,9,0,1)),$ & & &   \\
         &$\textcolor{blue}{((1,2,3,5,5),(7,5,6,5,5))}\}$ & & &   \\
         \hline 
    \end{tabular}
    }
    \caption{Steps taken by \toolname to prove the termination for the program in Fig.~\ref{fig:nest}}
    \vspace{-0.75cm}
    \label{tab:overview2}
\end{table}
\begin{gather*}
    \invar_o(m,y,z,n,k) \wedge \beta_o(m,y,z,n,k) \wedge \big(f(m,y,z,n,k) - f(m',y',z',n',k') < 1\big) \; \wedge \\
    \invar_i(x'') \; \wedge \; \neg \beta_i(x''), \quad \text{where} \quad  m' = m'', \ y' = y'' + 1, \ z' = z'' + 1, \ n' = n'' - 1, \ k' = k''  
    % \beta_o(m,y,z,n,k) = (y+n+1 \geq k^2 + z), \quad  \beta_i(m'',y'',z'',n'',k'') = (m''+y'' \geq 0 \wedge m'' \leq n''), \\  
    % \\  m'=2m+y, \ y'=z, \ z'=z+1, \ n'=n 
\end{gather*}
\noindent
Then, we refine the state of the algorithm. From the above validation query, we get a counterexample of the form $(p, p'', p')$. Here $p$ and $p'$ are the states at the beginning and end of one iteration of the outer loop, and $p''$ is the state at the exit of the inner loop. So, the refine phase has two options: (a) either refine the outer invariant $\invar_o$ first to exclude $p$ or (b) refine the inner invariant $\invar_i$ to exclude $p''$. If at least one of the invariants is refined, we check the validity of the ranking function again. If it is not possible to refine either of the invariants, we add $(p, p')$ to t. Once we choose an invariant to refine, we iteratively go through the Generate Invariant Sub-Phase and the Check Invariant Sub-Phase. Again, when trying to refine the invariants, if we find any states that are conclusively reachable, such as a state in $\ini_o$ at the beginning of the outer while loop, we use this state to guide the ranking function search in the future. If the invariant we chose cannot be refined, we try to refine the other invariant. Table~\ref{tab:overview2} gives the steps our algorithm takes to find the ranking function for the outer loop. This approach can be generalized to nested loops with arbitrary depth as explained in \S ~\ref{sec:theory} and \S ~\ref{sec:implementation}.

\subsection{Theoretical Guarantees}
While \toolname can be soundly instantiated with various templates, in \S~\ref{sec:theory}, we establish both relative completeness and efficiency guarantees if the template for ranking functions consists of a finite set of possibilities and the template for invariants is closed under intersection and does not have an infinite descending chain (with respect to the partial order defined in \S~\ref{sec:theory}). We further prove that even with a countably infinite ranking function template, \toolname will be able to find a proof of termination if it exists within the templates. 

\toolname's bi-directional feedback improves efficiency by simultaneously guiding the search for ranking functions and invariants. This approach is more efficient than traditional methods, which search for these components independently. Using a meet semilattice structure, we show that the worst-case bound for \toolname is the sum of the depths of the ranking function and invariant lattices, a significant improvement over naive feedback methods, which would have a worst-case bound of the product of the size of the individual lattices. This improvement can be attributed to the synergy between ranking function and invariant synthesis. The worst-case runtime analysis, detailed in \S~\ref{sec:efficiency}, shows that \toolname maintains efficiency while providing theoretical guarantees that are not provided by state-of-the-art algorithms~\cite{nuterm, dynamite, ultimate, aprove, verymax}.
\section{Annotated Programs}
\label{sec:annotatedpgm}
We consider programs in a simple while-programming language whose BNF grammar is given below, where, $x$ is a program variable, $e$ is a program expression, and $\beta$ is a Boolean expression.  
\[
\begin{array}{rcl}
\pro & ::= & \texttt{skip}\ |\ x \leftarrow e\ |\ \pro\:\text{;}\:\pro\ |\ \texttt{if } \beta \texttt{ then } \pro \texttt{ else } \pro\ |\ \texttt{while } \beta \texttt{ do } \pro \texttt{ od} 
\end{array}
\]
We use $\pgm,\pgm_1,\ldots$ to denote programs, $\loopfree,\loopfree_1,\ldots$ to denote loop-free programs, and $\looppgm,\looppgm_1,\ldots$ to denote loops, i.e., programs of the form $\texttt{while } \beta \texttt{ do } \pgm \text{ od}$.
The set of program states of $\pgm$ will be denoted as $\stat_{\pgm}$. The semantics of a program $\pgm$ denoted $\bracs{\pgm}$, is a binary relation on $\stat$ that captures how a program state is transformed when $\pgm$ is executed.

As explained in \S~\ref{sec:nestedloops}, invariants are used to approximate the reachable states for loops, the initial set of states for nested and sequential loops, and the semantics of loops. So, we define the concept of an \textit{annotated programs}, where each loop is associated with an inductive invariant. We then show how to use annotated programs to define the validity of ranking functions and define a partial order and a meet operation over annotated programs.  Using these definitions, in \S~\ref{sec:efficiency}, we ensure that \toolname refines the algorithm state in each iteration of the algorithm and is complete even for arbitrarily nested loops. 

Let us fix a collection $\tempi \subseteq \powerset{\stat}$ of \emph{invariants}, where $\stat$ is the set of program states. For technical reasons that we will emphasize later, we will assume that $\tempi$ is closed under intersection. That is, for any $\inv_1,\inv_2 \in \tempi$, $\inv_1 \cap \inv_2 \in \tempi$. Using the invariants in $\tempi$, we define \textit{annotated programs} as programs defined by the BNF grammar below, where $\inv$ is a member of $\tempi$.
\[
\begin{array}{rcl}
\anno & ::= & \texttt{skip}\ |\ x \leftarrow e\ |\ \anno\:\text{;}\:\anno\ |\ \texttt{if } \beta \texttt{ then } \anno \texttt{ else } \anno\ |\ \inv\: @\: \texttt{while } \beta \texttt{ do } \anno \texttt{ od} 
\end{array}
\]

We use $\annopgm,\annopgm_1, \ldots$ to denote annotated programs. For an annotated program $\annopgm$, its \emph{underlying program} can obtained by removing the invariant annotations on the $\texttt{while}$-loops and we denote this as $\core{\annopgm}$. This is defined in the Appendix~\ref{app:defns} (Definition~\ref{def:core-pgm}). For a (unannotated) program $\pgm$, the collection of all annotated programs over $\pgm$ will be denoted by $\anno(\pgm)$, i.e., $\anno(\pgm) = \setpred{\annopgm }{\core{\annopgm} = \pgm}$. We use $\alooppgm$ to denote an annotated subprogram of $\annopgm$ that corresponds to the loop $\looppgm$ in $\annopgm$. An annotated program, $\annopgm$, is \emph{correct} with respect to initial program states, $\ini$, if the annotations in $\annopgm$ are inductive invariants. This is denoted $\ini \models \annopgm$, and defined formally below. 

\begin{equation*}
% \label{def:correct-anno}
    \begin{aligned}
        \ini \models \loopfree \ &\triangleq \ \true \\
        \ini \models \texttt{if } \beta \texttt{ then } \annopgm_1 \texttt{ else }\annopgm_2 \ &\triangleq \ (\ini \cap \bracs{\beta} \models \annopgm_1) \wedge (\ini \cap \bracs{\neg\beta} \models \annopgm_2) \\
        \ini \models \annopgm_1\: ;\: \annopgm_2 \ &\triangleq \ (\ini \models \annopgm_1) \wedge (\ini\relcomp\bracs{\annopgm_1} \models \annopgm_2) \\
        \ini \models \inv\: @\: \texttt{while } \beta \texttt{ do }\annopgm_1 \texttt{ od} \ &\triangleq \ (\ini \subseteq \invar) \wedge ((\inv \cap \bracs{\beta})\relcomp\bracs{\annopgm_1} \subseteq \inv \times \inv) \wedge (\inv \cap\bracs{\beta} \models \annopgm_1) 
    \end{aligned}
\end{equation*}

Next, to define the semantics of an annotated program $\annopgm$ that captures the terminating computations of $\annopgm$, we assume a standard semantics of program expressions and Boolean expressions. We take $\bracs{\loopfree}$ to be the transition semantics of a loop-free program $\loopfree$ and $\bracs{\beta} \subseteq \stat$ to be the semantics of a Boolean expression $\beta$. The semantics of $\annopgm$ are defined inductively as follows.

\begin{equation*}
\label{eq:annosemantics}
    \begin{aligned}
        \bracs{\loopfree} \ &\triangleq \ \bracs{\loopfree} \\
        \bracs{\texttt{if } \beta \texttt{ then } \annopgm_1 \texttt{ else }\annopgm_2} \ &\triangleq \ \id{\bracs{\beta}} \relcomp \bracs{\annopgm_1} \cup \id{\bracs{\neg\beta}} \relcomp \bracs{\annopgm_2}\\
        \bracs{\annopgm_1\: ;\: \annopgm_2} \ &\triangleq \ \bracs{\annopgm_1} \relcomp\bracs{\annopgm_2} \\ 
        \bracs{\inv\: @\: \texttt{while } \beta \texttt{ do }\annopgm_1 \texttt{ od}} \ &\triangleq \ \inv \times (\inv \cap \bracs{\neg \beta}) 
    \end{aligned}
\end{equation*}

The difference between the semantics of an annotated program and that of a regular (unannotated) program lies in the case of loops: $\bracs{\inv\: @\: \texttt{while } \beta \texttt{ do }\annopgm_1 \texttt{ od}}$. We can over-approximate the semantics of the loop using the associated invariant $\invar$. The reachable states, $\reach$, at the entry of the while-loop, must satisfy the invariant. So, $\reach \subseteq \invar$. The states that satisfy $\beta$ enter the loop, while the states that do not satisfy $\beta$ exit. We can use the semantics of annotated programs to define when a ranking function proves the termination of an annotated $\texttt{while}$-loop, given an initial set of states.

\begin{definition}[Correctness of a Ranking Function]
\label{def:correct-rank}
    Let $\ini \subseteq \stat$ be a set of initial program states and $\rf: \stat \to \mathbb{R}$ be a candidate ranking function. We say that $\rf$ \emph{establishes} the termination of annotated program $\alooppgm = \inv\: @\: \texttt{while } \beta \texttt{ do } \annopgm \texttt{ od}$ from initial states $\ini$, i.e., $\ini \models \alooppgm, f$ if

    \begin{itemize}
        \item The annotations of $\alooppgm$ are correct with respect to $\ini$, i.e., $\ini \models \alooppgm$, and
        \item For all $(s,s') \in \id{\inv \cap \bracs{\beta}}\relcomp\bracs{\anno}$, 
        $\rf(s) \geq 0$ and $\rf(s) - \rf(s') \geq 1$
    \end{itemize}
\end{definition}

Note that in Definition~\ref{def:correct-rank}, we define the ranking function $\rf$ to be $\geq0$ and decrease by at least $1$ in each loop iteration. However, this can be generalized to general $\epsilon$ and $\delta$ (Lemma~\ref{aplemma:rfgen} in Appendix~\ref{appendix:proofs}).

Lastly, to reason about the completeness of \toolname when analyzing arbitrarily nested loops, we introduce a partial order ($\less$) and a meet operation ($\meet$) for the set of annotated programs ($\annopgm$). For annotated programs $\annopgm_1,\annopgm_2 \in \anno(\pgm)$, we say $\annopgm_2 \less \annopgm_1$ if for each loop in $\pgm$, the invariant annotating this loop in $\annopgm_2$ is contained in the corresponding invariant in $\annopgm_1$; the formal definition is deferred to Appendix~\ref{app:defns}. Next, $\annopgm_1 \meet \annopgm_2$ denotes the annotated program where each loop is annotated by the intersection of the invariants annotating the loop in $\annopgm_1$ and $\annopgm_2$. With these definitions, we establish the following proposition. The proof is in Appendix~\ref{appendix:proofs}:

\begin{proposition}
\label{prop:correct-refine}
    Let $\annopgm_1,\annopgm_2 \in \anno(\pgm)$ be two annotated programs. 
    \begin{enumerate}
        \item If $\ini \models \annopgm_2$ and $\ini \models \annopgm_1$, then $\ini \models \annopgm_1 \meet \annopgm_2$.
        \item If $\ini \models \annopgm_1,f$, $\ini \models \annopgm_2$, and $\annopgm_2 \less \annopgm_1$, then $\ini \models \annopgm_2,f$.
    \end{enumerate}
\end{proposition}

\section{Syndicate Framework}
\label{sec:theory}
The \toolname framework models bi-directional synergy, enabling the ranking function and invariant search to guide each other efficiently. The framework can be instantiated with \textit{templates} that define the set of possible ranking functions ($\tempr$) and invariants ($\tempi$). Each loop in a program can also be instantiated with a different template. We assume that the set of all possible program states ($\invar_{\true}$) is in $\tempi$ and that $\tempi$ is closed under intersection, i.e., if $\inv_1,\inv_2 \in \tempi$ then $\inv_1 \cap \inv_2 \in \tempi$. Consider a terminating program $\pgm$. The goal is to synthesize ranking functions from $\tempr$ and invariants from $\tempi$ that demonstrate the termination of every loop inside $\pgm$. To achieve this, we propose an algorithm that is parameterized by four sub-procedures: $\gc$, $\checkr$, $\gcounter$, and $\refine$. These sub-procedures can be instantiated or implemented using any method based on specific use cases, as long as they adhere to the semantics described in \S~\ref{sec:main-algorithm}. In \S~\ref{sec:implementation}, we present implementations for these sub-procedures that satisfy these semantics. In ~\S~\ref{sec:efficiency}, we also provide this algorithm's relative completeness and efficiency guarantees. 

\subsection{Algorithm State}
Consider a scenario where \toolname attempts to prove the termination of a loop defined as $\looppgm = \texttt{while } \beta \texttt{ do } \pgm_1 \texttt{ od}$ inside program $\pgm$ for a given set of initial states $\ini$. \toolname maintains an algorithm state that is used to guess a candidate ranking function and also verify its correctness. It then iteratively \textit{refines} the algorithm state until convergence. The algorithm state has two components: (1) an annotated program $\annopgm \in \anno(\pgm)$ such that $\ini \models \annopgm$; this can be easily initialized by labeling each loop in $\pgm$ with $\stat$, the set of all program states, (2) a finite set $\rp \subseteq \stat \times \stat$ such that for any correct annotation $\annopgm'$ of $\pgm$, $\rp \subseteq \id{\inv'\cap\bracs{\beta}}\relcomp\bracs{\annopgm'[\pgm_1]}$, where $\inv'$ is the annotation of $\looppgm$ in $\annopgm'$. The set $\rp$ can be initialized by executing $\pgm$ on randomly selected initial states from $\ini$ and collecting the execution traces. Since the pairs in $t$ are in the semantics of the loop body $\pgm_1$, it can be used to guess a candidate ranking function from the set $\func(\rp, \tempr) \triangleq \setpred{f \in \tempr}{\forall (s, s') \in \rp, f(s) \geq 0 \wedge f(s) - f(s') \geq 1}$, that contains all the functions that reduce and are bounded on the pairs in $t$. Next, the annotated program can be used to check the validity of the guessed candidate ranking functions.

\myparagraph{Refinement of the Algorithm State}
Checking the correctness of $\rf$ may lead to at least one of the following cases. (1) The identification of additional pairs of states $(s_i, s'_i) \notin t$ over which every valid ranking function must decrease. We add these newly discovered pairs to the algorithm state, i.e., $t' \leftarrow t \cup \{(s_i, s'_i)\}$. As a consequence, $\func(\rp', \tempr) \subset \func(\rp, \tempr)$. (2) The \textit{refinement} of the annotated program to $\annopgm'$ such that $\annopgm' \strctless \annopgm$. Thus, the algorithm progresses either by eliminating additional potential ranking functions or constructing a more precise model for the semantics of the program by a progressive refinement of the annotations of $\pgm$.

\begin{wrapfigure}{r}{0.3\linewidth}
    \centering
    % \vspace{-1.5em}
    \begin{tikzpicture}
    % Draw the oval shape
    \node[ellipse, draw, minimum width=2cm, minimum height=3cm] (oval) {};
    
    % Draw the points inside the oval
    \node[circle, fill=red!70!black, inner sep=2pt, label=above left:$\annopgm$] at (-0.4,0.5) {};
    \node[circle, fill=green!70!black, inner sep=2pt, label=below:$\annopgm'$] at (-0.1, -0.5) {};
    \node[circle, fill=green!70!black, inner sep=2pt, label=right:$\annopgm^*$] at (0.25, 0.8) {};
    \draw[->, decorate, decoration={snake, amplitude=0.5mm, segment length=5mm}] (-0.5,0.5) -- (-0.15, -0.4);
\end{tikzpicture}
    \caption{Assume that $\anno^*$ can prove the validity of $f$. Since $\anno' \less \anno$, $\anno$ can be refined to $\anno'$, which can also prove the validity of $f$.}
    % \vspace{-1.5em}
    \label{fig:lemma}
\end{wrapfigure}
Note that in our algorithm, the new annotation $\annopgm'$ is always $\less \annopgm$. An interesting and important observation at this point, described formally in Lemma~\ref{lem:meet}, is that if there exists an annotated program $\annopgm^*$ which can prove the validity of a ranking function $\rf$, then there also exists an annotated program $\annopgm' \less \annopgm$, which can prove $\rf$ valid (Fig.~\ref{fig:lemma}). This implies that when we refine the program's annotations, although we narrow down the space of all possible annotations, this does not eliminate the possibility of verifying the validity of any correct ranking function. This fact allows \toolname to descend the semi-lattice using any path and still make progress towards finding an annotated program with invariants that prove the validity of a ranking function. This is crucial in \toolname's correctness and efficiency. 

% An important observation at this point is that in our algorithm, the new annotation $\annopgm'$ is always $\less \annopgm$. In Lemma~\ref{lem:meet}, we prove that if there exists an annotated program $\annopgm^*$ which proves the validity of a ranking function $\rf$, then there also exists an annotated program $\annopgm' \less \annopgm$, which can prove $\rf$ valid (Fig.~\ref{fig:lemma}). Lemma~\ref{lem:meet} implies that when we refine the program's annotations, we narrow down the space of all possible annotations but do not eliminate the possibility of verifying the validity of any correct ranking function. This fact allows \toolname to descend the semi-lattice using any path and still make progress towards finding an annotated program with invariants that prove the validity of a ranking function. This helps in the algorithm's correctness and efficiency.

\begin{lemma}
\label{lem:meet}
Let $\annopgm,\annopgm^* \in \anno(\pgm)$ and $\rf$ be a ranking function. If $\ini \models \annopgm$ and $\ini \models \annopgm^*, \rf$ then $\exists \anno' \less \anno$ and $\ini \models \anno', \rf$.
\end{lemma}
\begin{proof}
    Let $\anno' = \anno \meet \anno^*$. Clearly, $\annopgm' \less \anno$. It follows from Proposition~\ref{prop:correct-refine} that $\ini \models \anno', f$. 
\end{proof}

\subsection{Parameterized Algorithm}
\label{sec:main-algorithm}
\begin{wrapfigure}{r}{0.53\textwidth}
% \vspace{-0.4cm}
    \begin{minipage}{\linewidth}
        \begin{algorithm}[H]
            \caption{\toolname}
            \label{general-algorithm}
            \begin{algorithmic}[1]
                \Procedure{\texttt{find\_ranking}}{$\rp, \annopgm$}
                \State $\texttt{refined} \gets \false$
                    \While{\true }
                        \If{$\neg\texttt{refined}$}
                            \State $\texttt{gen},f \gets \gc(\rp, \tempr)$
                            \If{$\neg \texttt{gen}$}
                                \State \Return $\false$
                            \EndIf
                        \EndIf
                        \If{$\checkr(f, \anno)$} 
                            \State \Return $\true$, $f$
                        \EndIf
                        \State $(p,p') \gets \gcounter(f, \annopgm)$
                        \State $\texttt{refined}, \; \annopgm, t \gets \refine(\annopgm, (p,p'), t, \tempi)$
                        \If{$\neg \texttt{refined}$}
                            \State $\rp \gets \rp \cup \{(p, p')\}$
                        \EndIf
                    \EndWhile 
                \EndProcedure
            \end{algorithmic}
        \end{algorithm}
    \end{minipage}
    % \vspace{-0.6mm}
\end{wrapfigure}

\toolname (Algorithm~\ref{general-algorithm}) iteratively generates a candidate ranking function and then uses it to refine the algorithm state if the candidate ranking function is not valid. When analyzing a program with multiple loops, Algorithm~\ref{general-algorithm} can be used on each loop in any order to prove the termination of every loop in the program. The algorithm is parameterized by four main sub-procedures: $\gc$, $\checkr$, $\gcounter$ and $\refine$ and the templates for ranking functions ($\tempr$) and invariants ($\tempi$). The function $\gc$ (line 5) should return a ranking function from the set $\func(\rp, \tempr)$ if one exists and should return \false if $\func(\rp, \tempr)$ is empty. The function $\checkr$ (line 10) should return \true if $\init \models \anno, f$ and \false otherwise. If $\checkr$ fails, the function $\gcounter$ should produce a \emph{counterexample} (line 13). The counterexample should be a pair of states $(p,p')$ such that (i) $(p,p') \in \bracs{\annopgm[\pgm_1]}$, (ii) $p \in \inv \cap \bracs{\beta}$, where $\inv$ is the invariant labeling $\looppgm$ in $\annopgm$, and (iii) either $\rf(p) < 0$ or $\rf(p) - \rf(p') < 1$. 

If $\rf$ is not a valid ranking function for $\looppgm$ then the counterexample $(p,p')$ is used to refine the program state (line 14). When $\pgm_1$ is loop-free, the counterexample can be eliminated in one of two ways: (1) The invariants in $\annopgm$ are refined so that $p$ is not an entry state of loop $\looppgm$ or $(p,p')$ is not in the semantics of the body $\pgm_1$, or (2) a new ranking function is synthesized that behaves correctly on the pair $(p,p')$. When $\pgm_1$ has loops, we define a sequence of states, $(p_0, p_1, \cdots, p_n)$ where $p_0 = p$, $p_n = p'$ and $p_1, \cdots, p_{n-1}$ are states occurring at the heads of the loops within $\pgm_1$. We then try to refine the invariants so that this sequence of states can no longer occur within the semantics of $\pgm_1$. It may still be the case that $(p,p')$ are in the semantics of $\pgm_1$ using the refined annotated program. In both the case of loop-free body and body containing inner loops, the counterexample $(p,p')$ is used to guide the refinement of $\anno$, demonstrating the \textit{first direction of feedback}: ranking function search guiding invariant(s) search. The crucial property required of $\refine$ is: 

\begin{enumerate}[noitemsep, nolistsep]
    \item If $\refine(\annopgm,(p,p'),\rp, \tempi)$ returns $(\true, \annopgm', t')$ then it must be the case $\annopgm' \strctless \annopgm$ and $\ini \models \annopgm'$.
    \item If $\refine(\annopgm,(p,p'),\rp, \tempi)$ returns $(\false, \annopgm', t')$ then $\annopgm' = \annopgm$ and for every correct $\annopgm''$ ($\ini \models \annopgm''$) with $\annopgm'' \strctless \annopgm$, $(p,p') \in \id{\inv'' \cap \bracs{\beta}}\relcomp\bracs{\annopgm''[\pgm_1]}$ where $\inv''$ is the invariant labelling $\looppgm$ in $\annopgm''$. In other words, $(p,p')$ is in the semantics of the body of the loop $\looppgm$ in every refinement of $\annopgm$.
\end{enumerate}

Implementing $\refine$ to satisfy these properties requires searching through the space of possible annotated programs, for example using a CEGIS loop. As we explain in \S~\ref{sec:sub-algorithms}, this involves generating candidate invariants that exclude $(p, p')$ and then checking their validity. While checking the validity of the generated invariants, $\refine$ can discover new reachable states and update $\rp$ to $\rp'$, including these states. Consequently, $t \subseteq t' \subseteq \bracs{\pgm_1}$ and $\dom(t') \subseteq \reach$. These newly discovered states in $\rp'$ can then guide the ranking function search towards potentially correct ranking functions, demonstrating \textit{the second direction of feedback}: invariant search guiding ranking function search.

Note that if $\refine$ is unable to refine (returns $\false$), then by the definition of $\refine$, for every $\annopgm'' \strctless \annopgm$, $(p,p')$ is in the semantics of the loop $\looppgm$ captured by  $\annopgm''$. The following lemma lifts this to all possible valid $\annopgm''$ and proves that if $\refine$ returns $\false$, it means that $(p,p')$ is in the semantics of the loop $\looppgm$ captured by \textit{all valid} annotated programs (not limited to $\strctless \annopgm$). So, we add $(p,p')$ to $\rp$ to get a new set that continues to satisfy the property that we required of $\rp$.

\begin{lemma}
    \label{lem:refine}
    Suppose $\refine(\annopgm,(p,p'),\rp, \tempi)$ returns $(\false, \annopgm, t')$. Consider any annotated program $\annopgm' \in \anno(\pgm)$ such that $\ini \models \annopgm'$. Let $\inv'$ be the annotation of $\looppgm$ in $\annopgm'$. Then $(p,p') \in \id{\inv'\cap\bracs{\beta}}\relcomp\bracs{\annopgm'[\pgm_1]}$.
\end{lemma}

\begin{proof}
    Suppose the claim is not true, i.e., for some correct annotated program $\annopgm'$, $(p,p') \not\in \id{\inv'\cap\bracs{\beta}}\relcomp\bracs{\annopgm'[\pgm_1]}$. Consider the program $\annopgm_1 = \annopgm \meet \annopgm' \less \annopgm$. As $\ini \models \annopgm$ and $\ini \models \annopgm'$,  it follows from Proposition~\ref{prop:correct-refine} that $\ini \models \annopgm_1$. Also, as $(p,p') \not\in \id{\inv'\cap\bracs{\beta}}\relcomp\bracs{\annopgm'[\pgm_1]}$ and $\annopgm_1 \less \annopgm'$, it follows from lemma \ref{aplemma:refinetrans} that $(p,p') \not\in \id{\inv_1 \cap \bracs{\beta}} \relcomp \bracs{\annopgm_1[\pgm_1]}$. So, $\refine$ can not return $\false$ as there is a correctly annotated program $\annopgm_1 \less \annopgm$ which does not contain $(p, p')$, leading to a contradiction.
\end{proof}

By processing counterexample $(p,p')$, \toolname always progresses. There are two cases to consider. If $\refine$ returns $\true$ then we have a new correct annotated program $\annopgm'$ that refines $\annopgm$. On the other hand, if $\refine$ returns $\false$ then we have new set $\rp_{\s{new}}$ with the property that $\func(\rp_{\s{new}},\tempr) \subset \func(\rp, \tempr)$. This is because $\rf \in \func(\rp, \tempr)\setminus\func(\rp_{\s{new}},\tempr)$ since $\rf$ is not a valid ranking function for the pair $(p,p')$. Thus, in this case $\func(\rp_{\s{new}},\tempr)$ becomes smaller.

\subsection{Soundness, Completeness, and Efficiency of \toolname}
\label{sec:efficiency}
Achieving efficiency or completeness independently in termination analysis is relatively straightforward, and many techniques exhibit either one. However, combining these properties, along with soundness, into a single framework is a challenge. While sound and complete algorithms can be devised by exhaustively exploring the search space, they tend to be inefficient due to the vast number of possibilities they need to explore. On the other hand, algorithms optimized for efficiency may sacrifice completeness or soundness by limiting the search space or relying on heuristics, potentially overlooking valid termination proofs. One of the contributions of this work is that \toolname achieves all three properties by employing a bi-directional feedback mechanism. 
% In this section, we describe these three properties in detail.

Algorithm~\ref{general-algorithm} can be shown to be sound for proving the termination of loop $\looppgm$ of $\pgm$ when instantiated with any ranking function and invariant templates if the four sub-procedures are sound. In this context, soundness means that if our algorithm succeeds and returns a ranking function $\rf$ from $\tempr$, then there exists an annotated program, $\anno$, using invariants in $\tempi$ such that $\ini \models \anno, \rf$.

\begin{theorem}[Soundness]
    \label{thm:mainsoundness}
Suppose there exist sound procedures for the functions $\gc$, $\checkr$, $\gcounter$, and $\refine$ that satisfy the properties defined in \S~\ref{sec:main-algorithm}. Then if Algorithm~\ref{general-algorithm} return $(\true, f)$, there exists $\anno$ using invariants from $\tempi$ such that $\ini \models \anno, \rf$.
\end{theorem}

\begin{proof}[Proof Sketch]
From the soundness of $\refine$, we can infer that the program state $\annopgm$ will always be correct, i.e., $\ini \models \annopgm$. Combined with the soundness of $\checkr(\rf, \annopgm)$, this ensures that any ranking function returned by the algorithm will be correct.
\end{proof}

Algorithm~\ref{general-algorithm} can also be shown to be complete for a subset of the possible ranking function and invariant templates, meaning that if there is a function $f \in \tempr$ and an annotated program, $\anno$, using invariants from $\tempi$ such that $f$ can be proved valid using $\anno$, then Algorithm~\ref{general-algorithm} will prove termination, and if there does not exist such an  $f$ and $\anno$, then Algorithm~\ref{general-algorithm} will terminate concluding that no ranking function and invariant can be found within the templates.

\begin{theorem}[Completeness]
    \label{thm:main}
    Suppose there exist sound and complete procedures for the functions $\gc$, $\checkr$, $\gcounter$, and $\refine$ that satisfy the properties defined in \S~\ref{sec:main-algorithm}. Further, suppose the set $\tempr$ is finite, and $\tempi$ is closed under intersection and has no infinite descending chains (w.r.t the subset relation). Then, if there exists $f \in \tempr$ and $\anno$ using invariants in $\tempi$ such that $\ini \models \anno, f$, Algorithm~\ref{general-algorithm} will return $(\true, f')$ for a valid $f'$, and otherwise, will return $\false$.
\end{theorem}

\begin{proof}[Proof Sketch]
    Suppose there is a ranking function $\rf^*$ and an annotated program $\annopgm^*$ such that $\annopgm^*$ is a correct annotation of $\pgm$ and $\rf^*$ is a valid ranking function that proves the termination of $\looppgm$. Let $(\annopgm,\rp)$ be the state of the algorithm at any point. Observe that because of the property that the algorithm maintains of $\rp$, we have $\rf^* \in \func(\rp,\tempr)$. By Lemma~\ref{lem:meet} there is a correct annotated program $\annopgm' (= \annopgm \meet \annopgm^*) \less \annopgm$ such that $\rf^*$ proves the termination of $\looppgm$ in $\annopgm'$. Finally, the assumptions on $\tempi$ and $\tempr$ mean that $\func(\rp,\tempr)$ and $\annopgm\downarrow = \setpred{\annopgm'}{\annopgm' \less \annopgm}$ are both finite sets and hence there can only be finitely many steps of refinement. Thus Algorithm~\ref{general-algorithm} will terminate with the right answer.
\end{proof}

\myparagraph{Ensuring Efficiency with Completeness} \toolname's bi-directional feedback ensures efficiency, even in instantiations that have the completeness guarantees stated above. Consider a meet semilattice ($\mathcal{L}_\tempr, \less, \meet)$ induced by $\tempr$ where each element is a set of ranking functions from $\tempr$ and $\less, \meet$ are defined by the set operations $\subseteq, \cap$. In each refinement step, we either (1) add a pair of states to $\rp$ creating $\rp_{new}$, which decreases the size of the set of ranking functions $\func$ determined by the algorithm's state, as $\func(\rp_{new},\tempr) \subset \func(\rp,\tempr)$. The number of times this refinement can be performed is bounded by the depth of the semilattice $\mathcal{L}_\tempr$, or (2) we refine the current annotation $\annopgm$ to $\annopgm'$, where $\annopgm' \strctless \annopgm$. This is bounded by the depth of the meet semilattice for the annotated programs $\mathcal{L}_\annopgm$. This bounds the number of iterations of our algorithm to $\mathsf{depth}(\mathcal{L}_\tempr) + \mathsf{depth}(\mathcal{L}_\annopgm)$. 

This worst-case bound is much better than an unguided search for invariants and ranking functions. A naive termination algorithm that searches for invariants and ranking functions independently would realize a worst-case bound of $|\mathcal{L}_\tempr| \ \times \ |\mathcal{L}_\annopgm|$, where $|\cdot|$ represents the number of elements in the lattice. Even a strategy that ensures refinement, a downward movement in the respective lattices, at each step of the algorithm, would have a worst-case bound of $\mathsf{depth}(\mathcal{L}_\tempr) \times \mathsf{depth}(\mathcal{L}_\annopgm)$. Thus, we argue that \toolname's synergistic synthesis for invariants and ranking functions has significant merit in improving runtimes. 

\myparagraph{Infinite Ranking Function Templates}
While the size of ranking function templates that satisfy the assumptions of Theorem~\ref{thm:main} can be very large, it may be useful to consider infinite sets of ranking functions. When given an infinite, recursively enumerable ranking function template, if the assumptions about the sub-procedures and the invariant template from Theorem~\ref{thm:main} hold, we can still ensure that \toolname can find a ranking function and an annotated program that proves the termination of a loop if such a ranking function and invariants exist in their respective templates.  For example, if $\tempr$ is the set of all linear ranking functions with integer coefficients, we can define a sequence of finite sets such that for all $f \in \tempr$, $f$ is in at least one of the finite sets. In this case, we can define sets $\tempr^{10}, \tempr^{100}, \tempr^{1000}, \cdots$, where $\tempr^n$ refers to the subset of $\tempr$ where the absolute value of each coefficient is bounded by $n$. Each of these sets is finite and every ranking function in $\tempr$ is included in at least one of the sets. We can then call Algorithm ~\ref{general-algorithm} on each set $\tempr^n$. Since each $\tempr^n$ is finite, Algorithm ~\ref{general-algorithm} will terminate for each set, either finding a ranking function and invariant or proving that none exist within their templates. This leads to an algorithm that will always find a ranking function and annotated program that proves termination if they exist given $\tempr$ and $\tempi$. This result holds for recursively enumerable ranking function templates because we can define a series of finite sets that include every ranking function in the template.

\section{Instantiation of  \toolname}
\label{sec:implementation}
In this section, we detail one instantiation of \toolname.  \S~\ref{sec:sub-algorithms} outlines the sound modeling of the four primary sub-procedures: $\gc$, $\checkr$, $\gcounter$, and $\refine$. These sub-procedures, based on the templates and programs being analyzed, can produce relatively complete algorithms for termination, as discussed in \S~\ref{sec:efficiency}. In \S~\ref{sec:spectrum}, we introduce new parameters specific to our instantiation that adjust the exploration of the ranking function search space versus the invariant search space. Tuning these parameters can enhance the algorithm's efficiency in practical applications. 

\subsection{Modeling sub-procedures using SMT} %for Complete Algorithm}
\label{sec:sub-algorithms}
\begin{wrapfigure}{r}{0.5\textwidth}
% \vspace{-1.0cm}
    \begin{minipage}{\linewidth}
        \begin{algorithm}[H]
    \caption{Algorithm for Refine State}\label{algorithm1}
    % \raggedright \textbf{Inputs:} Annotated program ($\annopgm$), counterexample $(p, p')$, set of traces ($t$), invariants template $\tempi$\\
    % \raggedright \textbf{Output:} $\true$ or $\false$, refined annotated program ($\annopgm'$), additional reachable states pairs ($t$)
    \begin{algorithmic}[1]    
        \Procedure{$\mathsf{refine}$}{$\anno, (p, p'), t, \tempi$}
            \State $\mathsf{inv_c}$ = \{ \}
            \State $\anno', \text{gen} \gets \gi(t, (p, p'), \mathsf{inv_c})$
            \While{\text{gen}}
                \If{$\checki(\anno')$} 
                    \State \Return $\true$, $\anno \meet \anno'$, $t$
                \EndIf
                \State $(c,c'), \mathsf{for\_pre} \gets \gcounterinv(\anno')$
                \If{$\mathsf{for\_pre}$}
                    \State $t \gets t \cup \{(c, c')\}$
                \Else
                    \State $\mathsf{inv_c} \gets \mathsf{inv_c} \cup \{(c, c')\}$ 
                \EndIf
                \State $\anno', \text{gen} \gets \gi(t, (p, p'), \mathsf{inv_c}, \tempi)$
            \EndWhile
            \State \Return $\false$, $\anno$, $t$
        \EndProcedure
    \end{algorithmic}
    \end{algorithm}
    \end{minipage}
    % \vspace{-0.2cm}
\end{wrapfigure}
In this instantiation of \toolname, we use SMT queries to check the soundness of ranking functions for a given loop. We define $\tempr$ as a template for ranking functions. For the function $\gc$, given $t$, we generate an SMT query encoding a symbolic ranking function in $\tempr$ that satisfies the reducing and bounded conditions for all of the pairs of states in $t$. Finding a satisfying assignment to this query is equivalent to finding a candidate ranking function. If there is no satisfying assignment, then $\func(\rp, \tempr)$ is empty, and $\gc$ returns $\false$.

We can also use an SMT query for a function that implements the semantics of both $\checkr$ and $\gcounter$. $\checkr$ checks the validity of $f$ given $\anno$. First, we define variables representing the state at the start of one iteration of the loop, $(x_1, \cdots, x_j)$, and add the condition $(x_1, \cdots, x_j) \in \invar$ where $\invar$ is the invariant of the loop we are analyzing, retrieved from $\anno$. Then, we encode the over-approximation of the loop body defined by $\anno$ into the SMT query.  For each loop in the body of the outer loop, we define new variables, $(x_{i,1}, \cdots, x_{i,j})$ to represent the states at the exit of each inner loop, adding the condition $(x_{i,1}, \cdots, x_{i,j}) \in \invar_i \cap \llbracket \neg \beta_i \rrbracket$, where $\invar_i$ and $\beta_i$ are the invariant and loop guard of the inner loop from $\anno$. We define variables representing the state at the end of the iteration, $(x'_1, \cdots, x'_j)$, and set these variables equal to the output state of $\bracs \anno$ when the input state is $(x_1, \cdots, x_j)$. Checking the validity of $f$ using $\anno$ is reduced to determining the satisfiability of the formula: $ (f(x_1, \cdots, x_j)-f(x'_1, \cdots, x'_j) < 1) \vee (f(x_1, \cdots, x_j) < 0) $. We return $\true$ if this formula is unsatisfiable, and $\false$ otherwise. In Algorithm~\ref{general-algorithm}, $\gcounter$ is only called when $\checkr$ returns $\false$. In this case, from the SMT query in $\checkr$, we have a satisfying assignment to $(x_1, \cdots, x_j)$ and $(x'_1, \cdots, x'j)$. Since we defined $\checkr$ to define symbolic states,$(x_{i,1}, \cdots, x_{i,j})$, at the end of each inner loop, instead of only returning $(p,p')$, $\checkr$ returns $(p_1, \cdots, p_m)$, where $p_1 = p$, $p_m = p'$ and $p_i = (x_{i,1}, \cdots, x_{i,j})$ for all $i \in [1,m]$. 

We now describe an algorithm that implements the semantics of $\refine$ defined in \S~\ref{sec:main-algorithm}. First, we consider a program with a single loop. Suppose we have a candidate ranking function, $f$, and a counterexample, $(p,p')$. In our algorithm state, we have an annotated program with one invariant, $\inv$. First we try to refine $\invar$ to $\invar'$ s.t. $\invar' \subseteq \invar \setminus \{p\}$. For this, we first try to find a valid invariant, $\invar''$, that excludes the state $p$ and includes all states $s$ such that $(s,s') \in t$. If this is possible, we return $\invar \cap \invar''$. Note that $\invar \cap \invar'' \in \tempi$ because $\tempi$ is closed under intersection, and $\invar \cap \invar'' \subset \invar$. To find $\invar''$, we iteratively generate possible invariants and check their validity until either we have found a valid invariant that excludes $p$ or proved that no such invariant exists. We show this procedure in Algorithm~\ref{algorithm1}. In lines 3 and 14, we generate new possible invariants. In line 5, we check the validity of the invariant. In lines 9-13, based on the output of the validity check we update the information used to generate new invariants. We use $inv_c$ to represent the set of pairs of states we obtained as counterexamples in previous iterations of checking the validity of possible invariants.
$inv_c$ is initialized to $\{\}$. We use the following SMT queries to find a candidate for $\invar''$ (\gi) and check its correctness (\checki).
\begin{equation}
\label{equation:invgen}
    \Big( p \notin \invar'' \Big) \wedge \Big( \bigwedge_{(s, s') \in t} s \in \invar'' \Big) \wedge\Big( \bigwedge_{(s,s') \in inv_c} s \in \invar'' \implies s' \in \invar'' \Big)
\end{equation}
\begin{equation}
\label{equation:invcheck}
    \Big( s \in \ini \wedge s \notin \invar'' \Big) \vee \Big( s \in \invar'' \wedge \llbracket \prog \rrbracket s \notin I'' \Big)
\end{equation}

If the SMT query checking the validity of $\invar''$ is satisfiable, $\gcounterinv$ returns a pair of states $(c, c')$, s.t. either $c \in \init \wedge (c,c') \not\in\invar''$ or $c \in \invar'' \wedge c' \notin I''$. In the former case, we add this state to $t$ (line 10). In the latter case, we add this state to $inv_c$ (line 12). With updated $t$ and $inv_c$, we repeat the process of generating and checking a new possible invariant. When $\refine$ returns the updated $t$, the counterexamples from the invariant search are used to guide the ranking function search.

% \myparagraph{Multiple Loops} 
In the case of multiple loops, we need to find $\anno' \strctless \anno$, where $\anno$ is the annotated loop program in our algorithm state. 
As discussed earlier, $\checkr$ will return $(p_1, \cdots, p_m)$ instead of only $(p,p')$. So, we have counterexamples, $p_i$, for each invariant, $\invar_i$, in the program. For each invariant, we maintain the set, $t_i$ and $inv_{c_i}$. $t_i$ is initialized from the initial program execution traces. $inv_{c_i}$ is initialized to an empty set. We use an SMT query to generate candidate invariants, $\invar'_i$ for all of the loops such that at least one of the loops in $\annopgm$ that excludes the counterexample, $p_i$ associated with $\invar_i$. Formally, we find a satisfying assignment to the constants in $\invar'_i$ such that 
\begin{equation}
\label{equation:invchecknested}
    \bigvee_i \Big( \neg p_i \in \invar_i \wedge \bigwedge_{(s,s') \in t_i} s \in \invar_i \wedge \big(\bigwedge_{(s,s') \in inv_{c_i}} s \in \invar_i \implies s' \in \invar_i \big) \Big)
\end{equation}

Once we find candidate invariants for all of the loops such that for at least one of the loops, $\invar'_i$ excludes $p_i$, we check the validity of the candidate invariants. This validity check is done the same way as in the single loop case, and the counterexamples are either added to $t_i$ or $inv_{c_i}$ accordingly. We repeatedly generate candidates until at least one of the invariants is refined or there does not exist any valid invariant for any of the loops that exclude the corresponding counterexample, $p_i$. In the latter case, we have shown that there is a sequence of reachable states w.r.t. the corresponding invariant templates that transition from $p$ to $p'$. Therefore, $(p,p')$ is in the semantics of the body of the loop for every possible refinement of $\anno$. In this case, we return $\false$. 

For the sub-procedures defined above to be complete, the set of all ranking functions in $\tempr$ and the bounded and reducing conditions on these ranking functions must be expressible in a decidable SMT theory.
Further, both the invariant generation SMT query and the validity check of candidate invariants have to be expressible in a decidable SMT theory. 

\subsection{Efficient Implementation through Adaptive Exploration of Search Spaces}
\label{sec:spectrum}
Although Algorithm~\ref{general-algorithm} is sound and complete with appropriate templates, it can be inefficient in practical settings. To address this, we introduce two parameters: (1) $\prefnum$, the maximum number of times $\refine$ is called for a given candidate ranking function, and (2) $\prefiter$, the maximum number of iterations within $\refine$. These parameters prevent \toolname from getting stuck in finding many invariants while attempting to prove the validity of an invalid ranking function or to prove a reachable state unreachable. Once this maximum number of calls to $\refine$ or iterations within $\refine$ is reached, \toolname will no longer try to refine the invariant. We show the code for this modification in Appendix~\ref{ap:genalgo} and ~\ref{ap:refalgo}. These parameters can be set by users depending on how much they want to explore the invariant search space compared to the ranking function search space.

When encountering a counterexample $(p,p')$ for the ranking function, if $\refine$ is not called or $\refine$ was not able to fully explore the invariant search space, we cannot conclude that $(p,p')$ is reachable. Instead of adding $(p,p')$ to $\rp$, we add it to a new set, $\rques$, which contains pairs of states that are not conclusively reachable or unreachable. Whenever we refine an invariant, we check if the current annotated program, $\anno$, can prove that any states in $\rques$ are unreachable and remove them if so. When generating candidate ranking functions, we use $\rp \cup \rques$. Thus, the generation phase outputs a function $f \in \tempr$ that decreases and is bounded on all states in $\rp \cup \rques$. If no candidate ranking function is found, the algorithm terminates. Since $\rques$ may contain unreachable pairs, this can lead to terminating without finding a valid ranking function even if one exists in $\tempr$. To maintain completeness, if we cannot find a valid ranking function, we ignore $\rques$ and run \toolname with $\prefnum$ and $\prefiter$ set to $\infty$, starting with the last $\rp$ and $\anno$. This approach leverages the information gained from the efficient analysis while preserving the algorithm's completeness guarantees.

\section{Evaluation}
\label{sec:evaluation}
We implemented \toolname using the instantiation and parameters introduced in \S~\ref{sec:implementation}. For ranking functions, we choose a commonly used template---lexicographic functions, $\langle e_1, \cdots, e_n \rangle$, where each $e_k = \sum_i\max(a^i_{0} + \sum_j a^i_{j}x_j,0)$. Here, $x_j$ represents program variables and $a^i_{j} \in \mathbb{Z}$. This includes lexicographic, linear, and non-linear ranking functions. There are two parameters when defining this template, $i$ and $n$. We define $T(i,n)$ to be a template in this form with $i$ summands per lexicographic expression and $n$ lexicographic expressions. Our implementation supports ranking function templates $T(i,n)$ for any  $i,n \in \mathbb{Z}$. However, in our evaluation, we use the following templates: $\tempr_1 = T(1,1), \; \tempr_2 = T(1,2), \; \tempr_3 = T(1,3), \; \tempr_4 = T(2,1), \text{and} \; \tempr_5 = T(2,2)$. For invariants, our implementation uses conjunctions of linear constraints,  $d_0 + \sum_jd_jx_j \ge 0$, where $d_j \in \mathbb{Z}$. 

The implementation supports integer programs and does not support function calls, except for calls to a random number generator. However, the implementation can be directly extended to handle arithmetic data types supported by SMT solvers, and non-recursive functions through inlining. For multiple loops, since Equation~\ref{equation:invchecknested} can be a large SMT query and is often not required to find useful invariants, we try to refine the invariants for the loops using individual SMT queries that exclude the state $p_i$ from $\invar_i$. If this is not possible, we add $(p,p')$ to the set $\rques$, described in Section~\ref{sec:spectrum}. We use a timeout of 120s. When computing the average time, we use 120s for instances that timeout.
We implemented \toolname in Python and used the Z3-Java API to check the validity of the ranking functions and invariants. All experiments are run on a 2.50 GHz 16 core 11th Gen Intel i9-11900H CPU with a main memory of 64 GB. 

\myparagraph{Benchmarks} We collect 168 benchmarks from $\nu$Term~\cite{nuterm}. These benchmarks are comprised of all of the terminating programs in the \textit{Term-crafted} set of SV-COMP~\cite{svcomp}, the \textit{AProVE\_09} set from TermComp~\cite{termcomp}, and \textit{$\nu$Term-advantage} from $\nu$Term~\cite{nuterm}. The SV-COMP and TermComp benchmarks are from the Software Verification Competition and Termination Competition respectively, which are used as a standard to judge state-of-the-art termination analysis tools. These benchmarks contain challenging integer programs, with non-linear assignments and loop guards, non-determinism, nested loops of maximum depth 3, and sequential loops. Integer programs constitute a strong dataset as they can be constructed by over-approximating the semantics of general, even heap-based programs~\cite{heap1, heap2, heap3} in a way that termination of the integer program implies the termination of the more general program. So, the termination of such integer programs is the standard evaluation metric for a termination analysis framework~\cite{ultimate, dynamite, nuterm, verymax, polykincaid, bits, popl23}. The programs, originally in C, were translated to Java by ~\cite{nuterm}. 
% The experiments for \toolname and $\nu$Term are run on the Java version while the experiments for \textsc{Ultimate}, VeryMax, AProVE, and DynamiTe use equivalent C benchmarks. 

As the primary evaluation in \S~\ref{sec:ablation}, we first demonstrate the advantages of using the novel bi-directional feedback between separate ranking functions and invariant searches. To this end, we implement baseline methods that either employ only a single direction of feedback, as used in ~\cite{bits, terminator, cooknew,datadrivenloop}, or directly search through the combined ranking function and invariant spaces, as used in ~\cite{verymax, popl23, chc}. Additionally, we evaluate several algorithms within \toolname that all use the bi-directional feedback approach and analyze the cases where \toolname fails. In the secondary evaluation in \S~\ref{sec:othertools}, we compare \toolname with state-of-the-art termination analysis tools~\cite{ultimate, aprove, nuterm, verymax, dynamite}.
% that are not limited to synthesizing a single ranking function for the given program. 
These tools rely on complex and sophisticated techniques, such as neural networks, multiple ranking functions tailored to different loop components, reduction orders, and optimizations for SMT solvers, specifically designed for the termination analysis queries of programs. 
Despite not using these advanced techniques, \toolname already has comparable or better performance than these tools, owing to its novel bi-directional feedback. Given that some of these techniques are orthogonal to \toolname, they can be combined with \toolname in the future to further improve the results. Finally, in \S~\ref{sec:case-studies}, we provide a deeper analysis of two representative benchmarks, demonstrating how \toolname quickly proves termination where other techniques fail, including a benchmark that none of the state-of-the-art tools could prove.

\subsection{Efficacy of the Synergistic Approach}
\label{sec:ablation}
As discussed in Section~\ref{sec:intro}, existing ranking function-based termination analysis tools can be broadly classified into three categories: (1) those that search for ranking functions and invariants completely independently~\cite{nuterm,dynamite}, (2) those that provide limited uni-directional feedback from one search to the other~\cite{bits, terminator, cooknew,datadrivenloop}, and (3) those that search for both the ranking function and the invariant through a single, monolithic query~\cite{verymax,popl23,chc}. We demonstrate that the bi-directional feedback employed in \toolname---where the invariant guides the ranking function and vice versa---leads to proving the termination of more benchmarks and using less average time than search strategies that use either direction of feedback alone or the single-query approach, where both components are synthesized together. Our experiments further reveal that the benefits of bi-directional feedback are consistent, regardless of the specific instantiation of the ranking function template. 

We implemented 3 baselines - \textbf{B-synergy1}, \textbf{B-synergy2}, and \textbf{B-combined}. B-synergy1 does not use the intermediate counterexamples found during the invariant search to guide the generation of new candidate ranking functions. B-synergy2 does not use the counterexamples from the ranking function search to guide the generation of new candidate invariants. B-combined searches for both a ranking function and invariants with a single monolithic query. We designed this query to match the form of the state-of-the-art termination analysis tool VeryMax~\cite{togethersmt}. In the following experiments, the baselines use the same templates as \toolname and only vary in how the ranking function and invariant searches interact. B-synergy1 and B-synergy2 are instantiated with the 5 ranking function templates $\tempr_i$ described in \S~\ref{sec:evaluation}. B-combined is only instantiated with $\tempr_1$ because B-combined does not scale well even with $\tempr_1$. We instantiate all of the baselines and \toolname with 0 initial traces. \toolname,  B-synergy1, and B-synergy2 all use $\prefnum = \prefiter = 10$.

\begin{figure}[]
    \centering
    \begin{minipage}{\linewidth}
        \centering
        \small
        \begin{tabular}{@{}c c c c@{}}
        \toprule 
        Ranking function Template & B-synergy1 & B-synergy2 & \toolname \\ 
        \midrule
        $\tempr_1$ & 76 & 84 & \textbf{98} \\ 
        $\tempr_2$ & 125 & 128 & \textbf{142} \\ 
        $\tempr_3$ & 132 & 127 & \textbf{148} \\ 
        $\tempr_4$ & 110 & 95 & \textbf{128} \\ 
        $\tempr_5$ & 115 & 113 & \textbf{134} \\
        \bottomrule
    \end{tabular}
    \captionof{table}{Number of benchmarks out of 168 proved by \toolname and baselines for different templates.}
    \label{table:baselines}
    \end{minipage}
    
    \begin{minipage}{\textwidth}
        \centering
        \begin{subfigure}{0.45\textwidth}
            \centering
            \includegraphics[width=\linewidth]{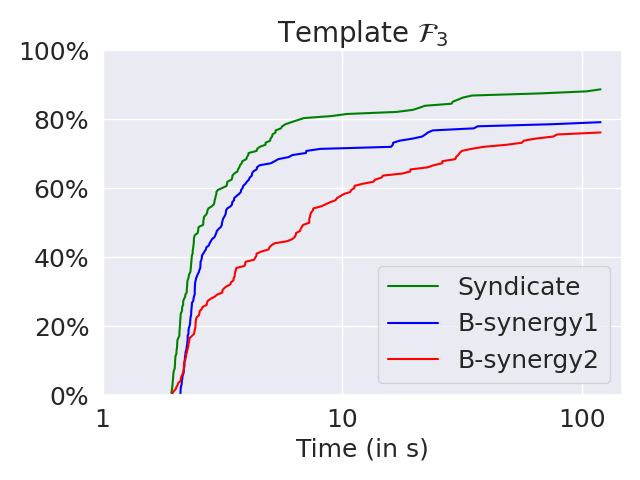}
            \caption{B-synergy1 and B-synergy2}
            \label{fig:baselines12}
        \end{subfigure}
        \hfill
        \begin{subfigure}{0.45\textwidth}
            \centering
            \includegraphics[width=\linewidth]{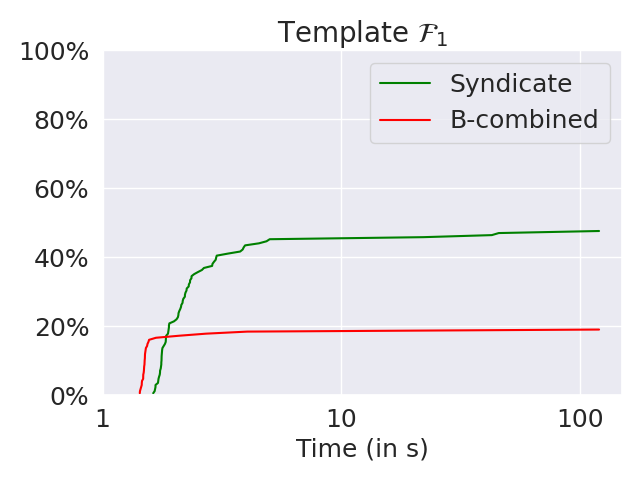}
            \caption{B-combined}
            \label{fig:baseline3}
        \end{subfigure}
        \caption{\% benchmarks proved with growing running time. Fig ~\subref{fig:baselines12} uses 168 benchmarks and Fig ~\subref{fig:baseline3} uses 144.}
        \label{fig:ablation}
        
    \end{minipage}

    \begin{minipage}{\textwidth}
        \centering
        \begin{subfigure}{0.45\textwidth}
        \centering
        \includegraphics[width=\linewidth]{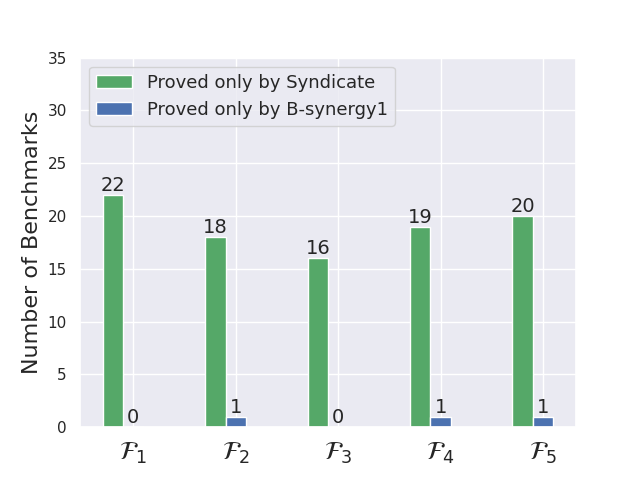}
        \caption{\toolname vs B-synergy1}
        \label{fig:bar-baseline1}
    \end{subfigure}
    \hfill
    \begin{subfigure}{0.45\textwidth}
        \centering
        \includegraphics[width=\linewidth]{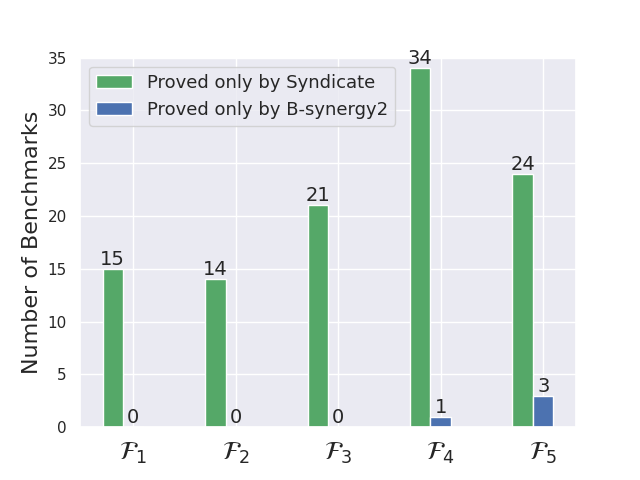}
        \caption{\toolname vs B-synergy2}
        \label{fig:bar-baseline2}
    \end{subfigure}
    \caption{No. of benchmarks proved by \toolname but not by the baselines out of 168 benchmarks.}
    \label{fig:bar-baseline}
    \end{minipage}
\end{figure}

In Fig.~\ref{fig:ablation}\subref{fig:baselines12}, we show the percentage of benchmarks proved with growing time for \toolname, B-synergy1, and B-synergy2 in the case of template $\tempr_3$. Both baselines perform worse than $\toolname$, indicating that guidance from both directions is vital for increasing the number of benchmarks proved within the timeout. The graphs for other templates are similar and can be found in Appendix~\ref{appendix:ablation}. In Table~\ref{table:baselines}, we show the numbers of benchmarks proved by \toolname, B-synergy1, and B-synergy2 using each template. There is up to 29\% improvement over B-synergy1 and up to 35\% improvement over B-synergy2.  Since these baselines use the same templates as \toolname, the improved performance is solely due to the searches guiding each other. In Fig.~\ref{fig:bar-baseline}, we show the number of benchmarks that are uniquely solved by B-synergy1 and B-synergy2 as compared to \toolname. We again see that \toolname is significantly better than both the baselines. We observe that in a few benchmarks, the baselines are better than \toolname due to the non-determinism involved in guessing the candidates. In Fig.~\ref{fig:ablation}\subref{fig:baseline3}, we compare the \toolname to B-combined, using $\tempr_1$ for both. Our implementation of B-combined only handles non-nested loops, so we use 144 benchmarks to evaluate B-combined. B-combined can solve only 31 of the total benchmarks. It solves these benchmarks very quickly because B-combined requires fewer iterations, but times out on relatively more complicated benchmarks because there is a large search space for the query. This shows that such a strategy does not scale as the complexity of the benchmarks grows.

\myparagraph{Variations within \toolname}
\begin{table}[]
    \centering
    \small
    \begin{tabular}{@{}lcccc@{}}
        \toprule
         & Parallel & Dynamic & Bounded & $\prefnum, \prefiter \neq \infty$ \\
         \midrule
         \textsc{\toolname-best} & $\checkmark$ & $\checkmark$ & $\checkmark$ & $\checkmark$ \\
         \textsc{\toolname-seq} & $\boldsymbol{\times}$& $\checkmark$ & $\checkmark$ & $\checkmark$ \\
         \textsc{\toolname-static} & $\checkmark$ & $\boldsymbol{\times}$ & $\checkmark$ & $\checkmark$ \\
         \textsc{\toolname-unbounded} & $\checkmark$ & $\checkmark$ & $\boldsymbol{\times}$ & $\checkmark$ \\
         \textsc{\toolname-full} & $\checkmark$ & $\checkmark$ & $\checkmark$ & $\boldsymbol{\times}$ \\
         \bottomrule
    \end{tabular}
    \caption{Different configurations of \toolname used in our evaluation.}
    \label{tab:variations}
\end{table}

\toolname allows many variations of analyses, all using our novel bi-directional feedback, each with different desirable properties for different types of benchmarks. We use representative implementations, introduced in Table~\ref{tab:variations}. \textsc{\toolname-seq} uses only one template, $\tempr_3$, and no parallelism. \textsc{\toolname-static} uses 0 initial traces for a fully static analysis. \textsc{\toolname-unbounded} uses ranking functions and invariants with no bounds on the coefficients, yielding infinite search spaces for ranking functions and invariants. \textsc{\toolname-full} sets $\prefnum=\prefiter=\infty$, producing a complete algorithm when analyzing a single loop. The best-performing version of \toolname in our evaluation is \textsc{\toolname-best}, which runs the 5 templates ($\tempr_1, \cdots, \tempr_5$) in parallel, starts with 100 initial traces, bounds the sum of the coefficients for each ranking function and invariant by $10000$, and uses $\prefnum = \prefiter = 10$. Table~\ref{table:versions} shows the number of benchmarks proved and the average time for the different algorithms within \toolname. 

\begin{wraptable}{r}{0.5\linewidth}
\centering
\small 
\begin{tabular}{@{}l@{\hspace{0mm}}r@{\hspace{3mm}}r@{}}
            \toprule 
            & $\#$ Proved  & Avg. Time (s) \\  
            \midrule
            \textsc{\toolname-best} & \textbf{151} & \textbf{15.37}  \\
            \textsc{\toolname-seq} & 147 & 18.35 \\
            \textsc{\toolname-static} & \textbf{151} & 15.91  \\
            \textsc{\toolname-unbounded} & 150 & 16.83 \\
            \textsc{\toolname-full} & 148 & 18.54  \\
            \bottomrule
        \end{tabular}
        \caption{Performance of variations of \toolname}
        \label{table:versions}
\end{wraptable}
All variations of \toolname rely on the same bi-directional feedback search strategy. In this evaluation, we removed or modified other features such as finite templates, parallelism, dynamic versus static analysis, and completeness to isolate the impact of the synergistic search strategy. Despite these modifications, all variations continue to perform well, demonstrating that the synergy in the search strategy is the key factor behind the strong performance.

% All variations of \toolname share the bi-directional feedback search strategy, showing that the good performance on these benchmarks is due to the synergistic search strategy rather than other factors like using finite templates, parallelism, dynamic vs static analysis, or compromising on completeness guarantees.

\myparagraph{Failure Analysis of \toolname}
We also analyze the programs that \textsc{\toolname-best} could not prove terminating within the timeout. Out of the benchmarks, 17 programs could not be proven to terminate. For 9 of these programs, there are no ranking functions and corresponding invariants in the templates we chose for our evaluation. We can conclude that there is no proof of termination for these templates because \textsc{\toolname-full}, where $\prefnum = \prefiter = \infty$ terminates within the timeout, proving the absence of a termination argument. For 4 of the remaining 8 programs, $\refine$ is unable to find an invariant that rules out the reachability of an unreachable state, so \toolname gets stuck in a loop generating new candidate invariants. For the other 4 benchmarks, \toolname keeps generating new possible ranking functions until the timeout.

\subsection{Comparison with State-of-the-Art Termination Analysis Tools}
\label{sec:othertools}
This section compares \toolname against various state-of-the-art termination analysis tools. Unlike \toolname, these tools leverage a wide array of sophisticated techniques, such as neural networks, multiple ranking functions tailored to different loop components, reduction orders, and optimizations for SMT solvers, specifically designed for the termination analysis queries of programs. Through our evaluation, we show that despite \toolname relying on a simpler termination proof approach---comprising a single ranking function for each loop and a set of invariants generated using a standard off-the-shelf SMT solver---it consistently performs on par with, or even better than, these more complex tools. It proves the termination of a broader range of benchmarks and reduces the average runtime by 16\% to 71\%.

\begin{table}[]
\centering
\resizebox{\textwidth}{!}{
\begin{tabular}{@{}l r r r r r@{}}
\toprule
\textbf{Tool Name} & \textbf{Ranking Function Generation} & \textbf{Invariant Synthesis} & \textbf{Main Technique} \\ 
\midrule
\multirow{2}{*}{\textbf{\textsc{Ultimate}}}   & Separate ranking function for & \multirow{2}{*}{Synthesized independently} & Learning termination \\ 
& each lasso-shaped sub-program &  & programs from traces  \\ \hline
\multirow{2}{*}{\textbf{VeryMax} }   & Multiple ranking functions & \multirow{2}{*}{Tightly coupled like in B-combined} & \multirow{2}{*}{SMT optimizations} \\ 
& for program subcomponents &  &  \\ \hline
\multirow{2}{*}{\textbf{AProVE}}     & Multiple ranking functions & \multirow{2}{*}{Synthesized independently} & Reduction orders,  \\ 
& for program subcomponents & & Dependency pairs  \\ \hline
\multirow{2}{*}{\textbf{DynamiTe}}   & \multirow{2}{*}{Single ranking function for each loop} & Uses external tool (DIG) & Ranking functions,  \\
& & for invariant synthesis & Recurrent sets  \\ \hline
\multirow{2}{*}{\textbf{${\nu}$Term}} & Neural network-based & \multirow{2}{*}{Enumerates invariants from a fixed set} & Neural ranking  \\ 
& ranking function generation & & function synthesis \\ \hline
\multirow{2}{*}{\textbf{\toolname}}  & \multirow{2}{*}{Single ranking function for each loop} & Synthesized independently with feedback & \multirow{2}{*}{Synergistic Synthesis} \\ 
&  &  from ranking function synthesis &  \\ 
\bottomrule
\end{tabular}
}
\caption{Characterization of termination analysis tools}
\label{table:characterization}
\end{table}

We compare \toolname against the following state-of-the-art tools for termination analysis characterized in Table~\ref{table:characterization}.
\textbf{\textsc{Ultimate}}~\cite{ultimate} decomposes programs and constructs ranking functions for these sub-programs. It then checks whether the current set of sub-programs captures the behavior of the entire loop. 
\textbf{VeryMax}~\cite{verymax} also combines many ranking functions to produce a termination argument, searching for different preconditions under which each ranking function is valid. This technique used for each ranking function is closely related to the B-combined baseline we described in the previous section. However, to make VeryMax work effectively, it introduces several complex optimizations for SMT solvers to ensure the generated queries are solvable. 
\textbf{AProVE}~\cite{aprove} is an ensemble of several proof techniques, including termination proofs through ranking functions, reduction orders, and dependency pairs.
\textbf{DynamiTe}\cite{dynamite} combines dynamic analysis with static strategies to infer ranking functions for termination and recurrent sets for non-termination. It uses traces and counterexamples similar to \toolname, but the invariant generation is outsourced to an external tool---DIG\cite{dig}. Despite using a state-of-the-art invariant synthesis tool, DynamiTe still cannot match \toolname's performance because, in \toolname, the invariants are directly guided by the ranking function and directly guide the ranking function search, allowing for a more efficient termination proof process. 
\textbf{${\nu}$Term}~\cite{nuterm} uses neural networks to learn a ranking function from program traces. While neural networks are powerful and efficient, \toolname still outperforms ${\nu}$Term due to the superior search strategy.

It is worth pointing out that \textsc{Ultimate}, VeryMax, and AProVE have previously won the Termination Competition~\cite{termcomp} and AProVE has also won the Software Verification Competition~\cite{svcomp}. All these tools employ parallel implementations in their approach. AProVE and VeryMax are fully parallelized, \textsc{Ultimate} and DynamiTe use parallelism when checking the validity of ranking functions, and both $\nu$Term and \toolname exploit parallelism to generate traces. In our evaluation, we allow full parallelism by using all available cores for these tools.

\begin{figure}
    \centering
    \begin{minipage}[t]{0.45\textwidth}
        \centering
        \small
        \vspace{-40mm}
        \begin{tabular}{@{}l@{\hspace{5mm}}r@{\hspace{5mm}}r@{\hspace{5mm}}}
         \toprule 
            & $\#$ Proved  & Avg. Time (s) \\  
            \midrule
            \textsc{\toolname-best} & \textbf{151} & \textbf{15.37}  \\
            \textsc{Ultimate} & 144 & 23.69 \\
            VeryMax & 143 & 18.30\\
            AProVE & 142 & 22.31\\
            DynamiTe & 126 & 53.04 \\
            $\nu$Term & $96^*$ & $33.13^*$ \\
            \bottomrule
        \end{tabular}
        \captionof{table}{Number of benchmarks out of 168 proved by \toolname and baselines for different templates. Other configurations of \toolname (Table~\ref{table:versions}) also prove more benchmarks within the timeout than other state-of-the-art tools. \textsc{\toolname-seq} matches or beats the average time per benchmark despite using no parallelism. $(^*)$The results for $\nu$Term are for 123 benchmarks.}
        \label{table:all}
    \end{minipage}
    \hfill
    \begin{minipage}[t]{0.48\textwidth}
        \centering
        \includegraphics[width=0.9\linewidth]{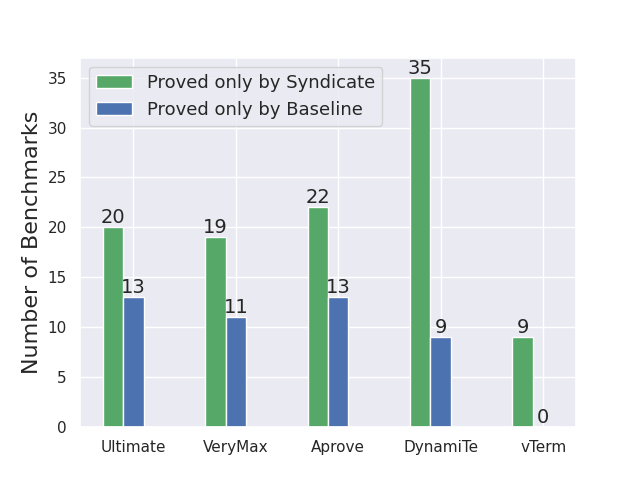}
        \caption{No. of benchmarks proved terminating by \toolname but not by the existing tool and vice-versa. 
        The total number of benchmarks used is 123 for comparison with $\nu$Term and 168 for all others.
        }
        \label{fig:comp-bar}
    \end{minipage}
\end{figure}

In Table~\ref{table:all}, we present the total number of benchmarks each technique was able to prove and the average time taken to run them. The efficacy of \toolname is demonstrated by its ability to prove the highest number of benchmarks in the shortest average time. \toolname achieves significant efficiency improvements over existing tools, reducing average runtime by 35.12\% compared to \textsc{Ultimate}, 16.01\% compared to VeryMax, 31.11\% compared to AProVE, and 71.02\% compared to DynamiTe. Notably, \toolname outperforms AProVE and VeryMax in runtime despite the latter two being fully parallelized and employing different termination arguments and SMT optimizations that could potentially complement \toolname's approach. Additionally, $\nu$Term does not handle sequential loops and calls to random number generators, supporting only 123 of the total benchmarks. Since we can exactly match the template of ranking functions used by $\nu$Term, when comparing \toolname to $\nu$Term, we have both tools use $\tempr_4$. $\nu$Term uses 1000 initial traces for its best performance. On these 123 benchmarks, $\toolname$ proves 105 benchmarks with an average time of 24.86 s and $\nu$Term proves 96 benchmarks with an average time of 33.13 s, yielding a 25\% improvement in runtime for $\toolname$.

From Tables~\ref{table:versions} and \ref{table:all}, it is evident that all variations of \toolname prove more programs terminating than the state-of-the-art techniques in Table~\ref{table:all}. Even \textsc{\toolname-seq}, which is run sequentially with only 1 template for ranking functions ($\tempr_3$), outperforms the existing tools that employ parallelism and are not limited to this template. We can also see that \toolname can efficiently handle infinite sets of possible ranking functions and invariants. In Fig.~\ref{fig:comp-bar}, we present the number of benchmarks that \textsc{\toolname-best} could prove terminating, but the other tools could not, and the number of benchmarks that each existing tool could prove terminating, and \textsc{\toolname-best} could not, within the timeout. This analysis highlights the unique benefits of \toolname's bi-directional synergistic approach compared to the existing state-of-the-art termination analysis tools.

\begin{figure}
    \centering
    \begin{subfigure}[t]{0.4\linewidth}
        \centering
        \begin{lstlisting}[numbers=left, basicstyle=\ttfamily\footnotesize, backgroundcolor=\color{backcolour}]
int x = y + 42;
while (x >= 0) {
    y = 2 * y - x;
    x = (y + x) / 2;
}
        \end{lstlisting}
        \caption{Not proved by AProVE or \textsc{Ultimate}}
        \label{lst:first}
    \end{subfigure}
    \hspace{0.1\linewidth}
    \begin{subfigure}[t]{0.4\linewidth}
        \centering
        \begin{lstlisting}[numbers=left, basicstyle=\ttfamily\footnotesize, backgroundcolor=\color{backcolour}]
int a = 0, b = 0;
while (a * b <= n || a * b <= m) {
    a = a + 1;
    b = b + 1;
}
        \end{lstlisting}
        \caption{Not proved by any existing tool}
        \label{lst:third}
    \end{subfigure}
    \caption{Benchmarks for Case Studies}
    \label{fig:casestudies}
\end{figure}

\subsection{Case Studies Showing Benefits of Bi-directional Feedback}
\label{sec:case-studies}
To shed more light on the efficacy of \toolname, we discuss two non-trivial programs from our benchmarks where some of the existing tools fail to prove termination within the timeout while \toolname succeeds in just 1.8s and 5.4s respectively. We compare more qualitatively against DynamiTe and $\nu$Term because they most closely resemble parts of \toolname. Consider the first program in Fig~\ref{fig:casestudies}. To prove its termination, it suffices to prove that $x$ reduces in successive iterations. A non-trivial algebraic rewriting reveals that, in each iteration, both $x$ and $y$ decrease by the initial value of $x-y$. One needs an invariant at least as strong as $x-y > 0$ to show that $x$ always reduces. AProVE and \textsc{Ultimate} are unable to prove this benchmark despite the termination argument being within their search space. \toolname proves this benchmark in just 1.8s owing to its synergistic search for a termination argument and results in $\rf^* \equiv \max(x+1, 0)$ and $\invar^* \equiv x-y\geq 1$. In this example, DynamiTe, which also relies on traces to come up with candidate ranking functions, is 5$\times$ slower and takes 11.9s to solve the problem. This is potentially because the invariant generation in DynamiTe is outsourced to DIG, and so is not guided by the ranking function search. Also, the traces DynamiTe uses to generate ranking functions are not guided by the invariant search. Invariants other than $\invar^*$ (like $x-y\geq42$ or even stronger) may also prove the validity of $\rf^*$ but a naive unguided search for invariants results in longer runtime. Further, $\nu$Term can also prove the termination in 2.56s, but its reliance on hard-coded invariants is not generalizable. 

The second program in Fig.~\ref{fig:casestudies} shows another non-trivial example consisting of a non-linear loop guard, for which none of the existing tools used in our experiments could prove termination. \toolname, however, generates a valid lexicographic ranking function $\langle \max(n-2m+a-4b+1, 0), \max(m+a-2b+6, 0) \rangle$ and loop invariant $a \geq b$ that proves that the generated ranking function is valid in just 5.4s. DynamiTe is unable to prove this benchmark, and our experiments with the invariant generation engine of DynamiTe (DIG) indicate that it does not produce invariants that are useful for proving the validity of the generated ranking functions within the timeout.  Further, even though there is a valid ranking function in the template used by $\nu$Term, it fails to find it because of its inability to guess the required invariants. These examples demonstrate that guiding both the invariant search and ranking function search the way \toolname does increase the number of benchmarks that can be proved terminating and significantly reduces the analysis time.

\section{Related Works}
\label{sec:related}
\textbf{Combined Search.} Numerous techniques have been designed to reason about program termination~\cite{nuterm, bits, terminator, loopster, dynamite, extremal, completelin, termpol, cooknew, togethersmt, verymax, polykincaid,popl23,chc}. Some of these techniques~\cite{chc, popl23,togethersmt,verymax,polykincaid} directly search for both the ranking function and over-approximation of the reachable states together by encoding both the transition relations of the program statements and properties of ranking functions in one logical formula, which can be prohibitively expensive to solve. Instead, \toolname separates the ranking function and invariant searches while exchanging sufficient bi-directional feedback to efficiently navigate through the space of possible ranking functions and invariants. 

\noindent
\textbf{Uni-directional Feedback.} Some termination analysis tools~\cite{dynamite, terminator, cooknew, bits, datadrivenloop, llms} call an external invariant generation engine or safety checker, which implicitly generates invariants after synthesizing candidate ranking functions. However, these methods typically provide no feedback to the ranking function search beyond the final invariants. So, the invariant search does not guide the ranking function search. Further, the external invariant synthesizer or safety checker has to effectively restart its search for invariants each time it is called because the invariant search state is not used to speed up future iterations of searching for invariants. \toolname, in contrast, uses its $\refine$ procedure to guide both subsequent iterations of $\refine$ itself and future iterations of generating candidate ranking functions. In \S~\ref{sec:ablation}, we show that having the invariant guide the ranking function search significantly increases the number of benchmarks proved within the timeout compared to only having one direction of feedback. Similar to how \toolname utilizes synergy between the invariants of multiple loops in a program and the ranking function of one of the loops, ~\cite{bitprecise} shares information about preconditions necessary for the termination of one function and the current invariants at the end of other functions for an inter-procedural analysis. While ~\cite{bitprecise} shares information when some code affects the initial set of other code, the use of annotated programs to define this synergy within \toolname allows it to have efficiency and completeness guarantees.

\noindent
\textbf{Formal Guarantees.} Some techniques for termination analysis provide completeness guarantees but handle a restricted class of programs\cite{extremal, completelin, completelin1}, such as linear ranking functions, which are not sufficient for more complicated programs. ~\cite{polykincaid} provides relative completeness guarantees for polyhedral ranking functions after statically computing a finite set of possible polyhedral ranking functions for each loop. However, \toolname provides relative completeness guarantees for more general types of non-linearity within ranking functions, including logical operators. ~\cite{chc} and ~\cite{popl23} provide relative completeness guarantees for a general class of ranking functions. However, neither of these approaches provides completeness and efficiency guarantees simultaneously. ~\cite{popl23} is empirically scalable because of its reliance on information from a parallel non-termination analysis. \toolname focuses only on termination and offers relative completeness guarantees while provably maintaining efficiency. \toolname also allows the user to adjust the analysis to be more efficient for specific use cases without sacrificing the relative completeness guarantees. 

\section{Conclusion}
\label{sec:conclusion}
We introduced a new general framework called \toolname for automated termination analysis based on bi-directional feedback between the invariant and ranking function searches. \toolname can be instantiated with different templates for invariants and ranking functions, is efficient, relatively complete, and can handle complex programs. \toolname can also be adapted to different practical settings by adjusting the relative exploration of the invariant and ranking function search spaces. We show that \toolname can prove significantly more benchmarks in less average time than state-of-the-art termination analysis techniques. \toolname is limited to the syntax defined in ~\S~\ref{sec:annotatedpgm} and does not handle arbitrary function calls, memory allocation, threads, or higher-order types. We leave extending the theory to handle these constructs as future work.

\clearpage
\section{Data-Availability Statement}
We have made our tool available in a VM that can be downloaded \href{https://zenodo.org/records/13822537?token=eyJhbGciOiJIUzUxMiJ9.eyJpZCI6ImZmMTU1MTV[…]KQqud0aqXjJ4XankUPoPwgSzHhnbn0toO2fIKDzrjcOhxWlz8WWEB6rR0AMr9sQ}{here}. This VM contains the logs for the experiments run in the paper in \texttt{\scriptsize{syndicate/src/paper\_results/}} and \texttt{\scriptsize{syndicate/src/paper\_results\_summary/}}, so the invariants and ranking functions found can be verified. The VM also contains the source code for the \toolname implementation, and allows the user to run the code on new Java benchmarks. It outputs the logs for the new Java benchmarks in \texttt{\scriptsize{syndicate/src/out/}}. We have also uploaded the same source code (and experimental logs) with a Dockerfile that can be downloaded \href{https://zenodo.org/records/13822873?token=eyJhbGciOiJIUzUxMiJ9.eyJpZCI6ImNlNTYxMGM[…]4kOQLcLfGViPfXbLX9ZVZzXDums_LlQiLYY1CizxUfiVDC-_-JrHtbpXUAi4iAA}{here} if the reviewer wants to use it in place of the VM. We have provided instructions for using this in a \texttt{\scriptsize{README}} file.
\bibliographystyle{ACM-Reference-Format}
\bibliography{main}
\clearpage
\appendix
\pagebreak
\section{Formal Definitions Omitted from Main Sections}
\label{app:defns}

% Recall that an \emph{inductive invariant} $\inv$ is a set of program states that satisfies the following two properties: (Initiation) $\inv$ contains all program states you could enter the loop in; (Consecution) If you run the loop body from a state in $\inv$, you end up in a state that also belongs to $\inv$. 
% Based on this, we can define what it means for all the annotations in a program to be correct. Notice that correctness needs to be defined w.r.t a set of initial states.
\begin{definition}
\label{def:anno-correct}
Let $\ini \subseteq \stat$ be a set of initial program states and let $\annopgm$ be an annotated program. $\annopgm$ is said to be \emph{correct} with respect to initial states $\ini$ if the annotations in $\annopgm$ are inductive invariants. This can be inductively defined based on the structure of $\annopgm$ as follows.
\begin{equation*}
\label{def:correct-anno}
    \begin{aligned}
        \ini \models \loopfree \ &\triangleq \ \true \\
        \ini \models \texttt{if } \beta \texttt{ then } \annopgm_1 \texttt{ else }\annopgm_2 \ &\triangleq \ (\ini \cap \bracs{\beta} \models \annopgm_1) \wedge (\ini \cap \bracs{\neg\beta} \models \annopgm_2) \\
        \ini \models \annopgm_1\: ;\: \annopgm_2 \ &\triangleq \ (\ini \models \annopgm_1) \wedge (\ini\relcomp\bracs{\annopgm_1} \models \annopgm_2) \\
        \ini \models \inv\: @\: \texttt{while } \beta \texttt{ do }\annopgm_1 \texttt{ od} \ &\triangleq \ (\ini \subseteq \invar) \wedge ((\inv \cap \bracs{\beta})\relcomp\bracs{\annopgm_1} \subseteq \inv \times \inv) \wedge (\inv \cap\bracs{\beta} \models \annopgm_1) 
    \end{aligned}
\end{equation*}

\end{definition}

Note that in the above definition, given an initial set for the annotated program $\alooppgm$, we use its constituent invariants to compute the initial set for the annotated subprograms of $\alooppgm$. The individual initial sets are consequently used for defining the correctness of the annotated programs.

\begin{definition}
\label{def:core-pgm}
    For an annotated program $\annopgm$, the \emph{underlying program} $\core{\annopgm}$ is the program where the annotations on the while loops have been removed. This is can be formally defined based on the structure of $\annopgm$ as follows.
    \begin{align*}
    \core{\alpha} &\triangleq \alpha \\
    \core{\texttt{if } \beta \texttt{ then } \annopgm{_1} \texttt{ else }\annopgm{_2}} &\triangleq \texttt{if } \beta \texttt{ then } \core{\annopgm{_1}} \texttt{ else } \core{\annopgm{_2}} \\
    \core{\annopgm{_1}\: ;\: \annopgm{_2}} &\triangleq \core{\annopgm{_1}}\: ;\: \core{\annopgm{_2}} \\
    \core{\inv \: @ \: \texttt{while } \beta \texttt{ do }\annopgm{_1} \texttt{ od}} &\triangleq \texttt{while } \beta \texttt{ do } \core{\annopgm{_1}} \texttt{ od} \\
\end{align*}
\end{definition}

\begin{figure}[htbp]
    \centering
        $
        \begin{array}{c}
            % \hspace{3cm}
            \inferrule*[lab={\footnotesize{\textsc{Atomic statements}}}]
            {
            % \consume{c, \alpha} = c'
            }
            {
            \alpha \less \alpha 
            } 
            \hspace{1cm}
            \inferrule*[lab={\footnotesize{\textsc{Sequence}}}]
            {
            \anno'_1 \less \anno'_2 \quad
            \anno''_1 \less \anno''_2
            }
            {
            \anno'_1 ; \anno''_1 \less \anno'_2 ; \anno''_2
            }  
            \\\\
            \inferrule*[lab={\footnotesize{\textsc{Loop}}}]
            {
            \invar_1 \subseteq \invar_2 \quad 
            \anno'_1 \less \anno'_2
            }
            {
            \Big(\invar_1 \ @ \ \texttt{while}(\beta) \texttt{ do }\anno'_1 \texttt{ od }\Big) \less \Big(\invar_2 \ @ \ \texttt{while}(\beta) \texttt{ do }\anno'_2 \texttt{ od }
            \Big)} 
            \\\\ 
            \inferrule*[lab={\footnotesize{\textsc{If}}}]
            {
            \anno'_1 \less \anno'_2 \quad
            \anno''_1 \less \anno''_2
            }
            {
            \Big(\texttt{if}(\beta) \texttt{ then } \anno'_1 \texttt{ else }\anno''_1\Big) \less \Big(\texttt{if}(\beta) \texttt{ then } \anno'_2 \texttt{ else }\anno''_2
            \Big)} 
        \end{array}
    $
    \caption{$\anno_1 \less \anno_2$}
    \label{fig:annoless}
    \label{fig:anno}
\end{figure}

\begin{figure}[htbp]
    \centering
    %     $
    %     \begin{array}{c}
    %         % \hspace{3cm}
    %         \inferrule*[lab={\footnotesize{\textsc{Atomic statements}}}]
    %         {
    %         % \consume{c, \alpha} = c'
    %         }
    %         {
    %         \alpha \less \alpha 
    %         } 
    %         \hspace{1cm}
    %         \inferrule*[lab={\footnotesize{\textsc{Loop}}}]
    %         {
    %         \invar_1 \subseteq \invar_2 \quad 
    %         \anno'_1 \less \anno'_2
    %         }
    %         {
    %         \Big(\invar_1 \ @ \ \texttt{while}(\beta) \texttt{ do }\anno'_1 \texttt{ od }\Big) \less \Big(\invar_2 \ @ \ \texttt{while}(\beta) \texttt{ do }\anno'_2 \texttt{ od }
    %         \Big)} 
    %         \\\\ 
    %         \inferrule*[lab={\footnotesize{\textsc{Sequence}}}]
    %         {
    %         \anno'_1 \less \anno'_2 \quad
    %         \anno''_1 \less \anno''_2
    %         }
    %         {
    %         \anno'_1 ; \anno''_1 \less \anno'_2 ; \anno''_2
    %         } 
    %         \hspace{1cm}
    %         \inferrule*[lab={\footnotesize{\textsc{If}}}]
    %         {
    %         \anno'_1 \less \anno'_2 \quad
    %         \anno''_1 \less \anno''_2
    %         }
    %         {
    %         \Big(\texttt{if}(\beta) \texttt{ then } \anno'_1 \texttt{ else }\anno''_1\Big) \less \Big(\texttt{if}(\beta) \texttt{ then } \anno'_2 \texttt{ else }\anno''_2
    %         \Big)} 
    %     \end{array}
    % $
    % \caption{$\anno_1 \less \anno_2$}
    % \label{fig:annoless}
    % \hfill
    % \begin{subfigure}{\textwidth}
    $
        \begin{array}{c}
            % \hspace{3cm}
            \inferrule*[lab={\footnotesize{\textsc{Atomic statements}}}]
            {
            % \consume{c, \alpha} = c'
            }
            {
            \alpha \meet \alpha = \alpha  
            } 
            \hspace{1cm}
            \inferrule*[lab={\footnotesize{\textsc{Sequence}}}]
            {
            \anno' = \anno'_1 \meet \anno'_2 \\\\
            \anno'' = \anno''_1 \meet \anno''_2
            }
            {
            \anno'_1 ; \anno''_1 \meet \anno'_2 ; \anno''_2 = \anno' ; \anno''
            } 
            \\\\
            \inferrule*[lab={\footnotesize{\textsc{Loop}}}]
            {
            \invar = \invar_1 \cap \invar_2 \quad 
            \anno = \anno'_1 \meet \anno'_2
            }
            {
            \Big(\invar_1 \ @ \ \texttt{while}(\beta) \texttt{ do }\anno'_1 \texttt{ od }\Big) \meet \Big(\invar_2 \ @ \ \texttt{while}(\beta) \texttt{ do }\anno'_2 \texttt{ od}\Big) = 
            \invar \ @ \ \texttt{while}(\beta) \texttt{ do }\anno \texttt{ od }
            } 
            \\\\ 
            \inferrule*[lab={\footnotesize{\textsc{If}}}]
            {
            \anno' = \anno'_1 \meet \anno'_2 \quad
            \anno'' = \anno''_1 \meet \anno''_2
            }
            {
            \Big(\texttt{if}(\beta) \texttt{ then } \anno'_1 \texttt{ else }\anno''_1 \Big)\meet \Big(\texttt{if}(\beta) \texttt{ then } \anno'_2 \texttt{ else }\anno''_2\Big) =           \texttt{if}(\beta) \texttt{ then } \anno' \texttt{ else }\anno''
            } 
        \end{array}
    $
    \caption{$\anno_1 \meet  \anno_2$}
    \label{fig:annomeet}
    % \end{subfigure}
    % \caption{Definitions of $\anno_1 \less \anno_2$ and $\anno_1 \meet \anno_2$.}
    % \label{fig:anno}
\end{figure}
\clearpage
\section{Lemmas and Proofs}
\label{appendix:proofs}
\noindent

\begin{lemma}
\label{aplemma:rfgen}
    If a ranking function $g$ for a loop is bounded from below by $\theta$ and decreases by at least $\delta (> 0)$ in each iteration, then there exists a ranking function $f$ for the loop that is bounded from below by 0 and decreases by at least 1 in each iteration.
\end{lemma}
\begin{proof}
    Consider $f(s)= (g(s) - \theta) / \delta$. Let $R$ be the set of reachable states at the start of the loop body and $T = \{(s_i, s_j)\}$ be the set of possible state pairs before and after the loop body.

    \begin{enumerate}
        \item \textbf{Bounded.}
        \setcounter{number}{1}
        \begin{align*}
            & \forall s \in R. \; g(s) \geq \theta &&&& \mycounter & \\
            & \forall s \in R. \; g(s) - \theta \geq 0 &&&& \mycounter & \\
            & \forall s \in R. \; (g(s) - \theta)/\delta \geq 0 & (as \; \delta > 0) &&& \mycounter & \\
            & \forall s \in R. \; f(s) \geq 0
        \end{align*}
        \item \textbf{Reducing.}
        \setcounter{number}{1}
        \begin{align*}
            & \forall (s_i, s_j) \in T. \; g(s_j) - g(s_i) \geq \delta &&&& \mycounter & \\
            & \forall (s_i, s_j) \in T. \; (g(s_j) - \theta) - (g(s_i) - \theta) \geq \delta &&&& \mycounter & \\
            & \forall (s_i, s_j) \in T. \; ((g(s_j) - \theta)/\delta) - ((g(s_i) - \theta)/\delta) \geq 1 &  (as \; \delta > 0) &&& \mycounter & \\
            & \forall (s_i, s_j) \in T. \; f(s_j) - f(s_i) \geq 1
        \end{align*}        
    \end{enumerate}
\end{proof}

\begin{lemma}
\label{aplemma:annos}
    If $\core{\annopgm{_1}} = \core{\annopgm{_2}}$, then $\bracs{\annopgm{_1} \meet \annopgm{_2}} = \bracs{\annopgm{_1}} \cap \bracs{\annopgm{_2}}$
\end{lemma}
\begin{proof}

Let $\core{\annopgm{_1}} = \core{\annopgm{_2}} = \prog$. We will do the proof by induction on the structure of $\prog$.\\

\noindent 
\textbf{Case 1 (Base Case): $\annopgm{_1} \equiv \annopgm{_2} \equiv \alpha$}
\setcounter{number}{1}
\begin{align*}
    & \annopgm{_1} \meet \annopgm{_2} = \alpha && && \mycounter & \\
    & \text{From (1), }\bracs{\annopgm{_1} \meet \annopgm{_2}} = \bracs{\alpha} &&&& \mycounter & \\
    & \bracs{\annopgm{_1}} = \bracs{\alpha} &&&& \mycounter & \\
    & \bracs{\annopgm{_2}} = \bracs{\alpha} &&&& \mycounter & \\
    & \text{From (3, 4), }\bracs{\annopgm{_1}} \cap \bracs{\annopgm{_2}} = \bracs{\alpha} &&&& \mycounter & \\
    & \text{From (2, 5), }\bracs{\annopgm{_1} \meet \annopgm{_2}} = \bracs{\annopgm{_1}} \cap \bracs{\annopgm{_2}}
\end{align*}

\noindent 
\textbf{Case 2: }$\annopgm{_1} \equiv \invar_1 \ @ \ \texttt{while }(\beta) \texttt{ do }\annopgm'_1 \texttt{ od} \qquad \annopgm{_2} \equiv \invar_2 \ @ \ \texttt{while }(\beta) \texttt{ do }\annopgm'_2 \texttt{ od}$
\setcounter{number}{1}
\begin{align*}
    & \annopgm{_1} \meet \annopgm{_2} = (\invar_1 \cap \invar_2) \ @ \ \texttt{while }(\beta) \texttt{ do } (\annopgm'_1 \meet \annopgm'_2) \texttt{ od} && && \mycounter & \\
    & \text{From (1), }\bracs{\annopgm{_1} \meet \annopgm{_2}} = (\invar_1 \cap \invar_2) \times (\invar_1 \cap \invar_2 \cap \bracs{\neg \beta}) &&&& \mycounter & \\
    & \bracs{\annopgm{_1}} = \invar_1 \times (\invar_1 \cap \bracs{\neg \beta}) &&&& \mycounter & \\
    & \bracs{\annopgm{_2}} = \invar_2 \times (\invar_2 \cap \bracs{\neg \beta}) &&&& \mycounter & \\
    & \text{From (3, 4), }\bracs{\annopgm{_1}} \cap \bracs{\annopgm{_2}} = (\invar_1 \cap \invar_2) \times (\invar_1 \cap \invar_2 \cap \bracs{\neg \beta}) &&&& \mycounter & \\
    & \text{From (2, 5), }\bracs{\annopgm{_1} \meet \annopgm{_2}} = \bracs{\annopgm{_1}} \cap \bracs{\annopgm{_2}}
\end{align*}

\noindent 
\textbf{Case 3: }$\annopgm{_1} \equiv \annopgm'_1 ; \annopgm''_1 \qquad \annopgm{_2} \equiv \annopgm'_2 ; \annopgm''_2$
\setcounter{number}{1}
\begin{align*}
    & \annopgm{_1} \meet \annopgm{_2} = (\annopgm'_1 \meet \annopgm'_2) ; (\annopgm''_1 \meet \annopgm''_2) && && \mycounter & \\
    & \text{From IH, }\bracs{\annopgm'_1 \meet \annopgm'_2} = \bracs{(\annopgm'_1} \cap \bracs{\annopgm'_2)} &&&& \mycounter & \\
    & \text{From IH, }\bracs{\annopgm''_1 \meet \annopgm''_2} = \bracs{(\annopgm''_1} \cap \bracs{\annopgm''_2)} &&&& \mycounter & \\
    & \bracs{(\annopgm'_1 \meet \annopgm'_2) ; (\annopgm''_1 \meet \annopgm''_2)} = \bracs{\annopgm'_1 \meet \annopgm'_2} \circ \bracs{\annopgm'_1 \meet \annopgm'_2} &&&& \mycounter & \\
    & \bracs{\annopgm{_1}} \cap \bracs{\annopgm{_2}} = (\bracs{\annopgm'_{1}} \circ \bracs{\annopgm''_{1}}) \cap (\bracs{\annopgm'_{2}} \circ \bracs{\annopgm''_{2}}) &&&& \mycounter & \\
    & \text{From (4, 5), }\bracs{\annopgm{_1} \meet \annopgm{_2}} = \bracs{\annopgm{_1}} \cap \bracs{\annopgm{_2}} 
\end{align*}

\noindent 
\textbf{Case 4: } $\annopgm{_1} \equiv \texttt{if}(\beta) \texttt{ then } \annopgm'_{1} \texttt{ else }\annopgm''_{1} \quad \annopgm{_2} \equiv \texttt{if}(\beta) \texttt{ then } \annopgm'_{2} \texttt{ else }\annopgm''_{2}$

% $\annopgm{_1} \equiv \mathtt{if }(\beta)\text{ then }\annopgm{'_1} \texttt{ else }\annopgm{''_1} \qquad \annopgm{_2} \equiv \mathtt{if }(\beta)\text{ then }\annopgm{'_2} \texttt{ else }\annopgm{''_2}$
\setcounter{number}{1}
\begin{align*}
    & \annopgm_1 \meet \annopgm_2 \equiv \texttt{if }(\beta)\texttt{ then }\annopgm'_1 \meet \annopgm'_2 \texttt{ else }\annopgm''_1 \meet \annopgm''_2 && && \mycounter & \\
    & \bracs{\annopgm{_1} \meet \annopgm{_2}} = \id{\bracs{\beta}}\circ\bracs{\annopgm'_1 \meet \annopgm'_2} \cup \id{\bracs{\neg\beta}}\circ\bracs{\annopgm''_1 \meet \annopgm''_2} &&&& \mycounter & \\
    & \text{From IH, }\bracs{\annopgm'_1 \meet \annopgm'_2} = \bracs{(\annopgm'_1} \cap \bracs{\annopgm'_2)} &&&& \mycounter & \\
    & \text{From IH, }\bracs{\annopgm''_1 \meet \annopgm''_2} = \bracs{(\annopgm''_1} \cap \bracs{\annopgm''_2)} &&&& \mycounter & \\
    & \bracs{(\annopgm{_1} \meet \annopgm{_2})} = \id{\bracs{\beta}}\circ(\bracs{(\annopgm'_1} \cap \bracs{\annopgm'_2)}) \cup \id{\bracs{\neg\beta}}\circ(\bracs{(\annopgm''_1} \cap \bracs{\annopgm''_2)}) &&&& \mycounter & \\
    & \bracs{\annopgm{_1}} = \id{\bracs{\beta}}\circ\bracs{(\annopgm'_1} \cup \id{\bracs{\neg\beta}}\circ\bracs{\annopgm''_1} &&&& \mycounter & \\
    & \bracs{\annopgm{_2}} = \id{\bracs{\beta}}\circ\bracs{(\annopgm'_2} \cup \id{\bracs{\neg\beta}}\circ\bracs{\annopgm''_2} &&&& \mycounter & \\
    % & \bracs{\annopgm{_1}} \cap \bracs{\annopgm{_2}} = (\bracs{\annopgm'_{1}} \circ \bracs{\annopgm''_{1}}) \cap (\bracs{\annopgm'_{2}} \circ \bracs{\annopgm''_{2}}) &&&& \mycounter & \\
    & \text{From (5, 6, 7), }\bracs{\annopgm_1 \meet \annopgm_2} = \bracs{\annopgm{_1}} \cap \bracs{\annopgm{_2}} &&&&
\end{align*}

\end{proof}
% \begin{lemma}
% ??
%     $s_1, s_2$ are sets and $t_1, t_2$ are relations, then\\
%     $(s_1 \circ t_1) \cap (s_2 \circ t_2) = ((s_1 \cap s_2) \circ (t_1 \cap t_2))$
% \end{lemma}

\begin{lemma}
\label{aplemma:less}
    If 
    \begin{enumerate}
        \item $S \models \annopgm$
        \item $S' \subseteq S$
    \end{enumerate}
    Then $S' \models \annopgm$ 
    % $S \models \annopgm$, then any $s'$ such that $s' \subseteq s$ satisfies $s' \models \annopgm$
\end{lemma}
\begin{proof}
We do this by induction on the structure of the underlying program $\core{\annopgm}$:\\
\textbf{Base Case:} $\annopgm = \alpha$\\
The proof holds trivially from the rules of atomic statements as $S' \models \alpha$ holds for all $S'$.

\textbf{Inductive Case:}\\
\textbf{Case 1: } $\annopgm \equiv \texttt{if}(\beta) \texttt{ then } \annopgm_{1} \texttt{ else }\annopgm_{2}$
\setcounter{number}{1}
\begin{align*}
    & S \models \annopgm && \text{Antecedant (1)} && \mycounter & \\
    & \text{From (1), } S \cap \bracs{\beta} \models \annopgm{_1} &&&& \mycounter & \\
    & S \cap \bracs{\neg \beta} \models \annopgm{_2} &&&& \mycounter & \\
    & S' \subseteq S && \text{Antecedant (2)} && \mycounter &\\
    & \text{From (2, 4), using IH, } S' \cap \bracs{\beta} \models \annopgm{_1} &&&& \mycounter & \\
    & \text{From (3, 4), using IH, } S' \cap \bracs{\neg \beta} \models \annopgm{_2} &&&& \mycounter & \\
    & \text{From (5, 6), using IH, } S' \models \texttt{if}(\beta) \texttt{ then } \annopgm_{1} \texttt{ else }\annopgm_{2} && \text{Consequent}&&  & \\
\end{align*}

\noindent
\textbf{Case 2: } $\annopgm \equiv \annopgm_{1} \; ; \; \annopgm_{2}$
\setcounter{number}{1}
\begin{align*}
    & S \models \annopgm && \text{Antecedant (1)} && \mycounter & \\
    & \text{From (1), } S  \models \annopgm{_1} &&&& \mycounter & \\
    & S \circ \bracs{\annopgm{_1}} \models \annopgm{_2} &&&& \mycounter & \\
    & S' \subseteq S && \text{Antecedant (2)} && \mycounter &\\
    & \text{From (2, 4), using IH, } S' \models \annopgm{_1} &&&& \mycounter & \\
    & \text{From (3, 4), using IH, } S' \circ \bracs{\annopgm{_1}}\models \annopgm{_2} &&&& \mycounter & \\
    & \text{From (5, 6), using IH, } S' \models \annopgm{_1} ; \annopgm{_2} && \text{Consequent}&&  & \\
\end{align*}

\noindent
\textbf{Case 3: } $\annopgm \equiv \invar \; @ \; \texttt{while}(\beta) \texttt{ do }\annopgm{_1} \texttt{ od}$
\setcounter{number}{1}
\begin{align*}
    & S \models \annopgm && \text{Antecedant (1)} && \mycounter & \\
    & \text{From (1), } S  \subseteq  \invar &&&& \mycounter & \\
    & \invar \cap \bracs{\beta} \circ \bracs{\annopgm{_1}} \subseteq  \invar &&&& \mycounter & \\
    & \invar \cap \bracs{\beta} \models \annopgm{_1} &&&& \mycounter & \\
    & S' \subseteq S && \text{Antecedant (2)} && \mycounter &\\
    & \text{From (2, 5), using IH, } S  \models \invar &&&& \mycounter & \\
    & \text{From (3, 4, 6), using IH, } S' \models \invar \; @ \; \texttt{while}(\beta) \texttt{ do }\annopgm{_1} \texttt{ od} && \text{Consequent}&&  & \\
\end{align*}
\end{proof}

\begin{lemma}
\label{aplemma:conj}
If
\begin{enumerate}
    \item $S \models \annopgm{_1}$
    \item $S \models \annopgm{_2}$
    \item $\core{\annopgm{_1}} = \core{\annopgm{_2}}$
\end{enumerate}
Then $S \models \annopgm{_1} \meet \annopgm{_2}$ 
\end{lemma}
\begin{proof}
$\core{\annopgm{_1}} = \core{\annopgm{_2}}$ ensures that the meet $\annopgm{_1} \meet \annopgm{_2}$ is defined as we define it only for annotated programs with the same underlying program.
Let $\core{\annopgm{_1}} = \core{\annopgm{_2}} = \prog$. We will do the proof by induction on the structure of $\prog$.\\
\textbf{Base Case:} $\annopgm{_1} = \annopgm{_2} = \alpha$\\
In this case, $\annopgm_3 = \annopgm{_1} \meet \annopgm{_2} = \alpha$. As we are given $S \models \annopgm{_1} (= \alpha)$, it follows that $S \models \annopgm_3 (= \alpha)$.

\noindent 
\textbf{Inductive Case:}\\
\textbf{Case 1: } $\annopgm{_1} = \texttt{if}(\beta) \texttt{ then } \annopgm_{11} \texttt{ else }\annopgm_{12} \quad \annopgm{_2} = \texttt{if}(\beta) \texttt{ then } \annopgm_{21} \texttt{ else }\annopgm_{22}$
\setcounter{number}{1}
\begin{align*}
    & S \models \annopgm{_1} && \text{Antecedant (1)} && \mycounter & \\
    & \text{From (1), } (S \cap \bracs{\beta})  \models \annopgm_{11} &&&& \mycounter & \\
    & (S \cap \bracs{\neg\beta})  \models \annopgm_{12} &&&& \mycounter & \\
    & S \models \annopgm{_2} && \text{Antecedant (2)} && \mycounter & \\
    & \text{From (4), } (S \cap \bracs{\beta})  \models \annopgm_{21} &&&& \mycounter & \\
    & (S \cap \bracs{\neg\beta})  \models \annopgm_{22} &&&& \mycounter & \\
    & \text{From (2, 5), using IH, } (S \cap \bracs{\beta})  \models \annopgm_{11} \meet \annopgm_{21} &&&& \mycounter & \\
    & \text{From (3, 6), using IH, } (S \cap \bracs{\neg\beta})  \models \annopgm_{12} \meet \annopgm_{22} &&&& \mycounter & \\
    & \text{From (7,8), } S \models \annopgm{_1} \meet \annopgm{_2} && \text{Consequent}&&  & \\
\end{align*}

\noindent
\textbf{Case 2: } $\annopgm{_1} = \annopgm_{11} \; ; \; \annopgm_{12} \quad \annopgm{_2} = \annopgm_{21} \; ; \; \annopgm_{22}$
\setcounter{number}{1}
\begin{align*}
    & S \models \annopgm{_1} && \text{Antecedant (1)} && \mycounter & \\
    & \text{From (1), } S  \models \annopgm_{11} &&&& \mycounter & \\
    & (S \circ \bracs{\annopgm_{11}}) \models \annopgm_{12} &&&& \mycounter & \\
    & S \models \annopgm{_2} && \text{Antecedant (2)} && \mycounter & \\
    & \text{From (4), } S  \models \annopgm_{21} &&&& \mycounter & \\
    & (S \circ \bracs{\annopgm_{21}}) \models \annopgm_{22} &&&& \mycounter & \\
    & \text{From (2, 5), using IH, } S  \models \annopgm_{11} \meet \annopgm_{21} &&&& \mycounter & \\
    & \text{From lemma~\ref{aplemma:annos}, }S \circ \bracs{(\annopgm_{11} \wedge \annopgm_{21})}) \equiv  S \circ (\bracs{\annopgm_{11}} \cap \bracs{\annopgm_{21}}) &&&& \mycounter & \\
    & \text{From (8), }S \circ \bracs{(\annopgm_{11} \wedge \annopgm_{21})}) \subseteq (S \circ \bracs{\annopgm_{11}}) &&&& \mycounter & \\
    & \text{From (8), }S \circ \bracs{(\annopgm_{11} \wedge \annopgm_{21})}) \subseteq (S \circ \bracs{\annopgm_{21}}) &&&& \mycounter & \\
    & \text{From (9), using Lemma~\ref{aplemma:less}, }S \circ \bracs{(\annopgm_{11} \wedge \annopgm_{21})}) \models \annopgm_{12} &&&& \mycounter & \\
    & \text{From (10), using Lemma~\ref{aplemma:less}, }S \circ \bracs{(\annopgm_{11} \wedge \annopgm_{21})}) \models \annopgm_{22} &&&& \mycounter & \\
    & \text{From (11, 12), using IH, } S \circ \bracs{(\annopgm_{11} \wedge \annopgm_{21})}) \models \annopgm_{12} \meet \annopgm_{22} &&&& \mycounter & \\
    & \text{From (7, 13), } S \models \annopgm{_1} \meet \annopgm{_2} && \text{Consequent}&&  & \\
\end{align*}

\noindent
\textbf{Case 3: } $\annopgm{_1} \equiv \invar_1 \; @ \; \texttt{while}(\beta) \texttt{ do }\annopgm_{11} \texttt{ od} \quad \annopgm{_2} \equiv\invar_2 \; @ \; \texttt{while}(\beta) \texttt{ do }\annopgm_{21} \texttt{ od}$\\
% $\annopgm_3 =  (\invar_1 \cap \invar_2) \; @ \; \texttt{while}(\beta) \texttt{ do }(\annopgm_{11} \meet \annopgm_{21}) \texttt{ od}$\\
\setcounter{number}{1}
\begin{align*}
    & S \models \annopgm{_1} && \text{Antecedant (1)} && \mycounter & \\
    & \text{From (1), } S  \subseteq  \invar_1 &&&& \mycounter & \\
    & \invar_1 \cap \bracs{\beta} \circ \bracs{\annopgm_{11}} \subseteq  \invar_1 &&&& \mycounter & \\
    & \invar_1 \cap \bracs{\beta} \models \annopgm_{11} &&&& \mycounter & \\
    & S \models \annopgm{_2} && \text{Antecedant (2)} && \mycounter & \\
    & \text{From (5), } S  \subseteq  \invar_2 &&&& \mycounter & \\
    & \invar_2 \cap \bracs{\beta} \circ \bracs{\annopgm_{21}} \subseteq  \invar_2 &&&& \mycounter & \\
    & \invar_2 \cap \bracs{\beta} \models \annopgm_{21} &&&& \mycounter & \\
    & \text{From (2, 6), }S \subseteq \invar_1 \cap \invar_2 &&&& \mycounter & \\
    & \text{From (3, 7), }(\invar_1 \cap \invar_2 \cap \bracs{\beta}) \circ \bracs{\annopgm_{11} \meet \annopgm_{21}} \subseteq (\invar_1 \cap \invar_2) &&&& \mycounter & \\
    & \text{From 4 using Lemma~\ref{aplemma:less}, }\invar_1 \cap \invar_2 \cap \bracs{\beta} \models \annopgm_{11} &&&& \mycounter \\
    & \text{From 8 using Lemma~\ref{aplemma:less}, }\invar_1 \cap \invar_2 \cap \bracs{\beta} \models \annopgm_{21} &&&& \mycounter \\
    & \text{From (10, 11, 12), using IH, } S \models \annopgm{_1} \meet \annopgm{_2} && \text{Consequent}&&  & \\
\end{align*}
\end{proof}

\begin{lemma}
\label{aplemma: meeter}
        If $\annopgm{_1}, \annopgm{_2} \in \lati(L, \tempi)$ satisfy:
        \begin{enumerate}
            \item $\ini \models \annopgm{_1}$
            \item $\ini \models \annopgm{_2}$
            \item $\tempi$ is closed under intersection
        \end{enumerate}
        then $\annopgm{_1} \meet \annopgm{_2} \in \lati(L, \tempi)$ and $\ini \models \annopgm{_1} \meet \annopgm{_2}$.
\end{lemma}

\begin{proof}
$\tempi$ being closed under intersection takes care of the fact that $\annopgm_3 = \annopgm{_1} \meet \annopgm{_2}$ is in $\lati(L, \tempi)$ as the meet operator takes the intersection of invariants in $\annopgm{_1}$ and $\annopgm{_2}$. Now, as $\ini  \models \annopgm{_1}$ and  $\ini  \models \annopgm{_2}$, $\ini \models \annopgm{_1} \meet \annopgm{_2}$ follows directly from the lemma~\ref{aplemma:conj}.
\end{proof}

\begin{lemma}
\label{aplemma:refinetrans}
     If $\core{\annopgm{_1}} = \core{\annopgm{_2}}$, then $\annopgm{_1} \less \annopgm{_2} \implies \bracs{\annopgm{_1}} \subseteq \bracs{\annopgm{_2}}$
\end{lemma}
\begin{proof}
Let $\core{\annopgm{_1}} = \core{\annopgm{_2}} = \prog$. We will do the proof by induction on the structure of $\prog$.\\
% We do this by induction on the structure of $\annopgm{_1}$.

    \textbf{Case 1 (Base case): $\annopgm{_1} \equiv \annopgm{_2} \equiv \alpha$}
    \setcounter{number}{1}
    \begin{align*}
        & \bracs{\annopgm{_1}} = \alpha && && \mycounter & \\
        & \bracs{\annopgm{_2}} = \alpha &&&& \mycounter & \\
        & \text{From (1, 2), }\bracs{\annopgm{_1}} \subseteq \bracs{\annopgm{_2}} &&&& \mycounter & \\
    \end{align*}
        \textbf{Case 2: }$\annopgm{_1} \equiv \invar_1 \ @ \ \texttt{while }(\beta) \texttt{ do }\annopgm'_1 \texttt{ od} \qquad \annopgm{_2} \equiv \invar_2 \ @ \ \texttt{while }(\beta) \texttt{ do }\annopgm'_2 \texttt{ od}$
    \setcounter{number}{1}
    \begin{align*}
        & \text{From } (\annopgm{_1} \less \annopgm{_2}), \; \invar_1 \subseteq \invar_2 && && \mycounter & \\
        & \bracs{\annopgm{_1}} = \invar_1 \times (\invar_1 \cap \bracs{\neg \beta}) &&&& \mycounter & \\
        & \bracs{\annopgm{_2}} = \invar_2 \times (\invar_2 \cap \bracs{\neg \beta})  &&&& \mycounter & \\
        & \text{From (1, 2, 3), }\bracs{\annopgm{_1}} \subseteq \bracs{\annopgm{_2}} &&&& \mycounter & \\
    \end{align*}
    \textbf{Case 3: }$\annopgm{_1} \equiv \annopgm'_1 ; \annopgm''_1 \qquad \annopgm{_2} \equiv \annopgm'_2 ; \annopgm''_2$
    \setcounter{number}{1}
    \begin{align*}
        & \text{From } (\annopgm{_1} \less \annopgm{_2}), \; \annopgm'_1 \less \annopgm'_2 \wedge \annopgm''_1 \less \annopgm''_2 && && \mycounter & \\
        & \bracs{\annopgm{_1}} = \bracs{\annopgm'_1} \circ \bracs{\annopgm''_1}   &&&& \mycounter & \\
        &\bracs{\annopgm{_2}} = \bracs{\annopgm'_2} \circ \bracs{\annopgm''_2} &&&& \mycounter & \\
        & \text{From 1 using IH, }\bracs{\annopgm'_1} \subseteq \bracs{\annopgm'_2} &&&& \mycounter & \\
        & \text{From 1 using IH, }\bracs{\annopgm''_1} \subseteq \bracs{\annopgm''_2} &&&& \mycounter & \\
                & \text{From (2,3,4,5), }\bracs{\annopgm{_1}} \subseteq \bracs{\annopgm{_2}} &&&& \mycounter & \\
        & 
    \end{align*}

\noindent
\textbf{Case 4: } $\annopgm{_1} \equiv \texttt{if}(\beta) \texttt{ then } \annopgm'_{1} \texttt{ else }\annopgm''_{1} \qquad \annopgm{_2} \equiv \texttt{if}(\beta) \texttt{ then } \annopgm'_{2} \texttt{ else }\annopgm''_{2}$
\setcounter{number}{1}
\begin{align*}
        & \text{From } (\annopgm{_1} \less \annopgm{_2}), \; \annopgm'_1 \less \annopgm'_2 \wedge \annopgm''_1 \less \annopgm''_2 && && \mycounter & \\
        & \bracs{\annopgm{_1}} = \id{\bracs{\beta}} \circ \bracs{\annopgm'_1} \cup \id{\bracs{\neg\beta}} \circ \bracs{\annopgm''_1}   &&&& \mycounter & \\
        & \bracs{\annopgm{_2}} = \id{\bracs{\beta}} \circ \bracs{\annopgm'_2} \cup \id{\bracs{\neg\beta}} \circ \bracs{\annopgm''_2}   &&&& \mycounter & \\
        & \text{From 1 using IH, }\bracs{\annopgm'_1} \subseteq \bracs{\annopgm'_2} &&&& \mycounter & \\
        & \text{From 1 using IH, }\bracs{\annopgm''_1} \subseteq \bracs{\annopgm''_2} &&&& \mycounter & \\
                & \text{From (2,3,4,5), }\bracs{\annopgm{_1}} \subseteq \bracs{\annopgm{_2}} &&&& \mycounter & \\
        & 
\end{align*}
\end{proof}

\begin{lemma}
\label{aplemma:name}
    If
    \begin{enumerate}
        \item $\ini \models \alooppgm{_1}, f$
        \item $\ini \models \alooppgm{_2}$
        \item $\alooppgm{_2} \less \alooppgm{_1}$
    \end{enumerate}
    then $\ini \models \alooppgm{_2},f$
\end{lemma}
\begin{proof}
    Recall that $\alooppgm$ is an annotated loop program $\alooppgm = \inv\: @\: \texttt{while } \beta \texttt{ do } \annopgm \texttt{ od}$. From definition~\ref{def:correct-rank}  
    As defined, $\ini \models \alooppgm{_2},f$ holds when $\ini \models \alooppgm{_2}$ (already given) and $f$ is valid for the loop program $\alooppgm{_2} = \invar_2 \; @ \; \texttt{while}(\beta) \texttt{ do }\annopgm'_{2} \texttt{ od}$. We are given that $f$ is valid for loop program $\alooppgm{_1} = \invar_1 \; @ \; \texttt{while}(\beta) \texttt{ do }\annopgm'_{1} \texttt{ od}$, i.e. it satisfies $$\forall (s,s') \in \id{\inv_1 \cap \bracs{\beta}}\relcomp\bracs{\annopgm{_1}}, \rf(s) \geq 0 \text{ and } \rf(s) - \rf(s') \geq 1$$.

    Also, $\annopgm{_2} \less \annopgm{_1} \implies (\invar_2 \subseteq \invar_1) \wedge (\annopgm'_2 \less \annopgm'_1)$, which is equivalent to $(\invar_2 \subseteq \invar_1) \wedge (\bracs{\annopgm'_2} \subseteq \bracs{\annopgm'_1})$ from lemma~\ref{aplemma:refinetrans}. 
    This implies that $(\id{\inv_2 \cap \bracs{\beta}}\relcomp\bracs{\annopgm{_2}}) \subseteq (\id{\inv_1 \cap \bracs{\beta}}\relcomp\bracs{\annopgm{_1}})$.
    So, from the definition of $f$ being valid for $\alooppgm{_1}$, we can also say:

    $$\forall (s,s') \in \id{\inv_2 \cap \bracs{\beta}}\relcomp\bracs{\annopgm{_2}}, \rf(s) \geq 0 \text{ and } \rf(s) - \rf(s') \geq 1$$.
    
    This proves that $\ini \models \alooppgm{_2},f$.
\end{proof}

Proof for Proposition~\ref{prop:correct-refine}:
\begin{enumerate}
    \item $\ini \models \annopgm_2$ and $\ini \models \annopgm_1$, then $\ini \models \annopgm_1 \meet \annopgm_2$ follows from lemma ~\ref{aplemma: meeter}. 
    \item If $\ini \models \annopgm_1,f$, $\ini \models \annopgm_2$, and $\annopgm_2 \less \annopgm_1$, then $\ini \models \annopgm_2,f$ follows from Lemma~\ref{aplemma:name}
\end{enumerate}

\clearpage
\section{Ablation Studies}
\label{appendix:ablation}

\begin{figure}[H]
    \centering
    \begin{subfigure}{0.45\textwidth}
        \includegraphics[width=\linewidth]{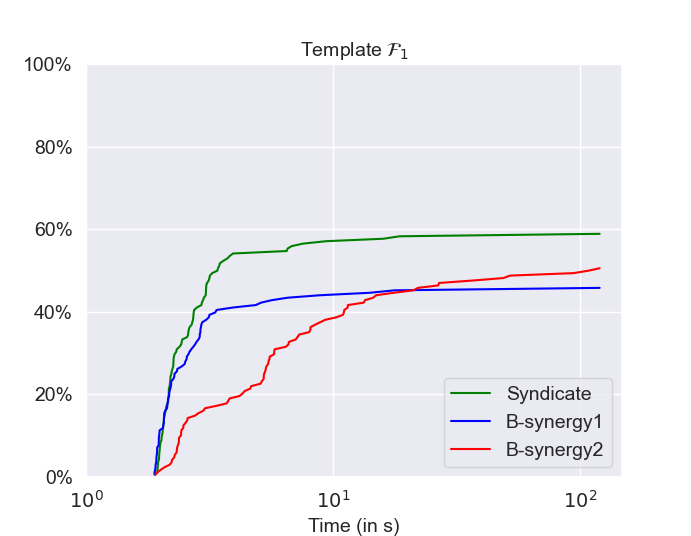}
        \caption{Template $\tempr_1$}
        \label{fig:baselines12a}
    \end{subfigure}
    \begin{subfigure}{0.45\textwidth}
        \includegraphics[width=\linewidth]{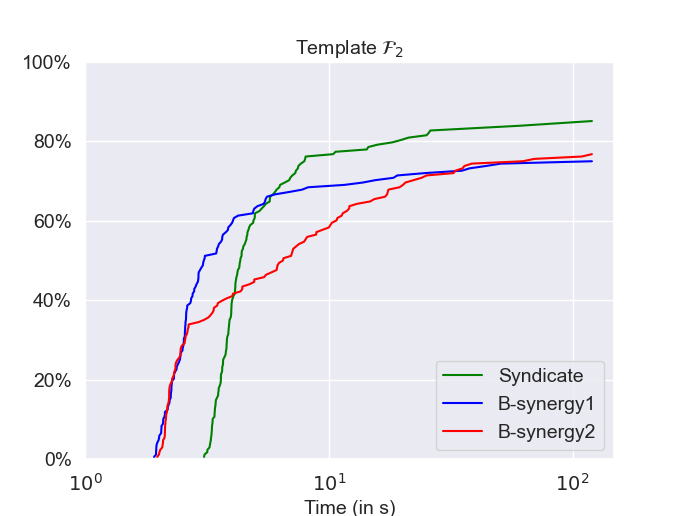}
        \caption{Template $\tempr_2$}
        \label{fig:baseline12b}
    \end{subfigure}
    \begin{subfigure}{0.45\textwidth}
        \includegraphics[width=\linewidth]{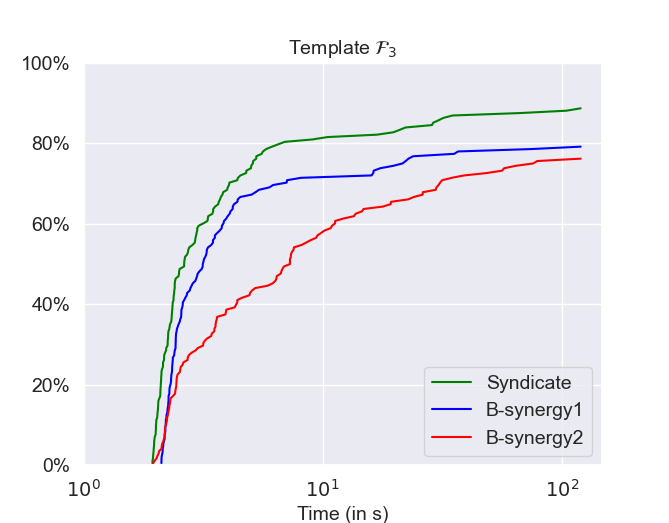}
        \caption{Template $\tempr_3$}
        \label{fig:baselines12c}
    \end{subfigure}
    \begin{subfigure}{0.45\textwidth}
        \includegraphics[width=\linewidth]{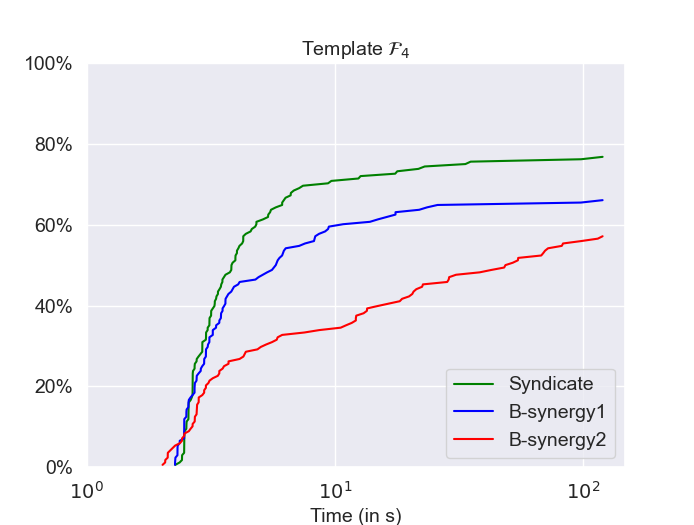}
        \caption{Template $\tempr_4$}
        \label{fig:baseline12d}
    \end{subfigure}
    \begin{subfigure}{0.45\textwidth}
        \includegraphics[width=\linewidth]{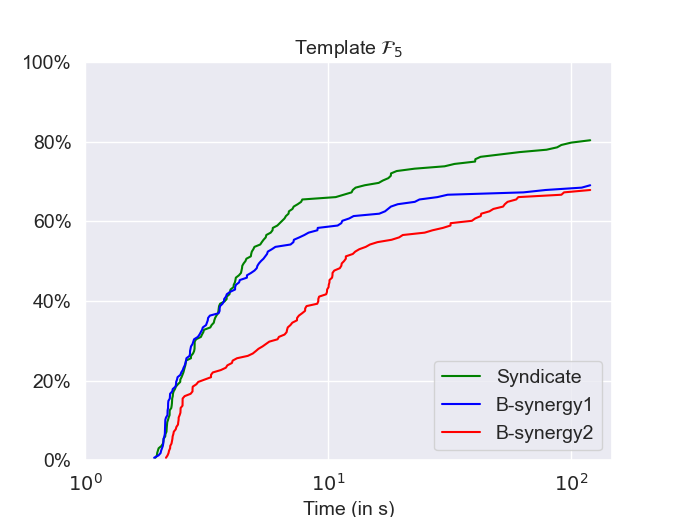}
        \caption{Template $\tempr_5$}
        \label{fig:baseline12e}
    \end{subfigure}
    \caption{Percentage of benchmarks proved terminating with growing running time.}
\end{figure}

\clearpage
\section{Implementation SMT Queries}
\label{appendix:smtqueries}

\textbf{$\gc$ -} This function returns an element of the set $\func(\rp, \tempr)$. Since every function in $\tempr$ is bounded below by 0, this problem reduces to solving the following SMT query:
\begin{align*}
\bigwedge_{(s,s') \in \rp} \Big(\big(\sum_i \max(a^i_{0} + \sum_j a^i_{j}x_j, 0)\big) - \big(\sum_i \max(a^i_{0} + \sum_j a^i_jx'_j, 0)\big) \geq 1 \Big) 
% \bigwedge_{(X,X') \in \rp} \Big(\max(y_0 + \sum_i y_iX_i, 0) + \max(z_0 + \sum_i z_iX_i, 0) \Big) - \\
% \Big(\max(y_0 + \sum_i y_iX'_i, 0) + \max(z_0 + \sum_i z_iX'_i, 0) \Big) \ge 1
\end{align*}
Here, $s = (x_0, \cdots x_j)$ represents a program state with concrete valuations of all program variables, and $a^i_j$ are variables representing the coefficients of the ranking function. \\

\textbf{$\checkr$ -} This function takes as input $f$ and $\anno$. As discussed in \S~\ref{sec:sub-algorithms}, using the over-approximation of the loop body defined by $\bracs \anno$, we define variables representing the state at the start of one iteration, $(x_1, \cdots, x_j)$, the end of the iteration,  $(x'_1, \cdots, x'_j)$, and at the exit of each loop within the body of the outer loop, $(x_{1,j}, \cdots, x_{i,j})$. Then, we determine the satisfiability of the formula: $ f((x_1, \cdots, x_j))-f((x'_1, \cdots, x'_j)) < 1 $. If this formula is unsatisfiable, then we have proven the validity of $f$ because we have already established that $f$ is bounded. If this formula is satisfiable, then the algorithm continues to refine either $\anno$ or $\rp$.\\

\textbf{$\gcounter$ -} From the SMT query in $\checkr$, we have found a satisfying assignment to $(x_1, \cdots, x_n)$ and $(x'_1, \cdots, x'_n)$. $\gcounter(f, \anno)$ returns the counter-example $(p,p')$, where $p = (x_1, \cdots, x_n)$ and $p' = (x'_1, \cdots, x'_n)$. In our implementation, since we defined $\checkr$ to define variables, $(x_{1,j}, \cdots, x_{n,j})$, for the states at the end of each inner loop, instead of only returning $(p,p')$, this function can return $(p_1, \cdots, p_m)$, where $p_1 = p$, $p_m = p'$ and $p_i = (x_{1,j}, \cdots, x_{i,j})$ for all $i \in [1,m]$. This sequence of states will be used to define $\refine$ to satisfy the properties defined in \S~\ref{sec:theory}.\\

\textbf{$\refine$ -}
As described in \S~\ref{sec:sub-algorithms}, when considering a program with one loop, we use Equations~\ref{equation:invgen} and ~\ref{equation:invcheck} to generate and check the validity of candidate invariants that exclude $(p_1, \cdots, p_m)$ from the over-approximation of the transition relation of the loop. When considering multiple loops, we try to refine the invariants for the loops in individual SMT queries to exclude the state $p_i$ from $\invar_i$. If we cannot refine the invariants of any of the loops to exclude the corresponding counter-example, we assume that the states $(p_1, \cdots, p_m)$ are reachable.

% New Algo Try

\clearpage
\section{Parameterized Algorithms}
\label{appendix:newalgos}

\begin{figure}[H]
    \begin{algorithm}[H]
            \caption{\toolname Parametrized Version}
            \label{ap:genalgo}
            \raggedright\textbf{Inputs:} $\rp$ - Set of traces, $\annopgm$ - Program annotated with the invariants (initialized with $\invar_{\true}$) \\
            \raggedright\textbf{Output:} $\true$ if termination is proved, $\false$ otherwise
            \begin{algorithmic}[1]
                \Procedure{$\fr$}{$\pgm, \ptraces, \prefnum, \prefiter$}
                \State $\rp \gets \gtraces(\pgm, \ptraces)$
                \State $\rques \gets \{\}$
                \State $\annopgm \gets \pgm$ annotated with $\invar_{\true}$ for all loops.           
                \State $\texttt{should\_gen\_rf} \gets \true$
                \State $\texttt{refine\_per\_rf} \gets 0$
                    \While{\true }
                        \If{$\texttt{should\_gen\_rf}$}
                            \State $\texttt{gen},f \gets \gc(\rp \cup  \rques)$
                            \State $\texttt{refine\_per\_rf} \gets 0$
                            \If{$\neg \texttt{gen}$}
                                \State \Return $\false$
                            \EndIf
                        \EndIf
                        \If{$\checkr(f, \anno)$} 
                            \State \Return $\true$
                        \EndIf
                        \State $(p,p') \gets \gcounter(f, \annopgm)$
                        \State $\texttt{is\_refined}, \; \annopgm, t, \texttt{reached\_lim}  \gets \refine(\annopgm, (p,p'), t, \prefiter)$                        
                        % \State $\rp \gets \rp \cup \{(q_i,q'_i)\}$
                        \If{$\neg \texttt{is\_refined}$}
                            \If{$\texttt{reached\_lim}$}
                                \State $\rques \gets \rques \cup \{(p, p')\}$
                                % \Comment{$(p, p')$'s reachability is uncertain; refine reached its limit.}
                            \Else
                                \State $\rp \gets \rp \cup \{(p, p')\}$
                            \EndIf
                        \EndIf
                        \State $\texttt{refine\_per\_rf} \gets \texttt{refine\_per\_rf} + 1$
                        \State $\texttt{should\_gen\_rf} \gets \neg \texttt{is\_refined} \; || \; (\texttt{refine\_per\_rf} == \prefnum)$

                        % \State \Return $\mathsf{find\_rank}$($\rp, \anno, f, \neg \texttt{is\_refined}$)
                    \EndWhile 
                \EndProcedure
            \end{algorithmic}
        \end{algorithm}
\end{figure}

\begin{figure}
    \begin{algorithm}[H]
    \centering
% \sgcom{PLACING}
    \caption{Algorithm for Refine State}\label{ap:refalgo}
    \raggedright\textbf{Inputs:} Annotated program ($\annopgm$), counter-example ($(p, p'$), reachable state pairs ($t$)\\
    \raggedright\textbf{Output:} $\true$ or $\false$, refined annotated program ($\annopgm'$), additional reachable states pairs ($t$)
    \begin{algorithmic}[1]    
\Procedure{$\mathsf{refine}$}{$\anno, (p, p'), t, \prefiter$}
    \State $\mathsf{inv_c}$ = \{ \}
    % \State $r = \dom(\rp)$ 
    % \State $t' = []$
    \State $\anno', \text{gen} \gets \gi(t, (p, p'), \mathsf{inv_c})$
    \State $\texttt{iter} \gets 1$
    \State $\texttt{reached\_lim} \gets \false$
    \While{\text{gen}}
        \If{$\checki(\anno')$} 
            \State \Return $\true$, $\anno \meet \anno'$, $t$
        \EndIf
        \State $(c,c'), \mathsf{for\_pre} \gets \gcounterinv(\anno')$
        \If{$\mathsf{for\_pre}$}
            \State $t \gets t \cup \{(c, c')\}$
        \Else
            \State $\mathsf{inv_c} \gets \mathsf{inv_c} \cup \{(c, c')\}$ 
        \EndIf
        \If{$\texttt{iter} == \prefnum$}
            \State $\texttt{reached\_lim} \gets \true$
            \State \textbf{break}
            \Comment{Break if we have reached the invariant generation limit.}
        \EndIf
        \State $\anno', \text{gen} \gets \gi(t, (p, p'), \mathsf{inv_c})$
    \EndWhile
    \State \Return $\false$, $\anno$, $t$, $\texttt{reached\_lim}$
\EndProcedure
\end{algorithmic}
\end{algorithm}
\end{figure}

\end{document}